\pdfoutput=1
\documentclass[11pt]{article}

\usepackage[margin=1in]{geometry}
\usepackage{amsmath,amssymb,mathtools}
\usepackage{cases}
\usepackage{amsthm}
\usepackage{booktabs}
\usepackage{enumitem}
\usepackage[dvipsnames]{xcolor}
\usepackage{tikz}
\usepackage{microtype}
\usepackage[hidelinks]{hyperref}
\usepackage[nameinlink,noabbrev]{cleveref}

\hypersetup{
  pdftitle={Minimum Cardinalities of Multipartite Unextendible Product Bases },
  pdfsubject={A complete classification of minimum-cardinality unextendible product bases},
  pdfkeywords={unextendible product basis, bound entanglement, orthogonality graph, Fourier matrix, perfect matching}
}

\newtheorem{theorem}{Theorem}[section]
\newtheorem{lemma}{Lemma}[section]
\newtheorem{proposition}{Proposition}[section]

\newtheorem{definition}{Definition}[section]

\newtheorem{claim}{Claim}

\crefname{theorem}{theorem}{theorems}
\Crefname{theorem}{Theorem}{Theorems}
\crefname{lemma}{lemma}{lemmas}
\Crefname{lemma}{Lemma}{Lemmas}
\crefname{proposition}{proposition}{propositions}
\Crefname{proposition}{Proposition}{Propositions}
\crefname{corollary}{corollary}{corollaries}
\Crefname{corollary}{Corollary}{Corollaries}
\crefname{definition}{definition}{definitions}
\Crefname{definition}{Definition}{Definitions}
\crefname{remark}{remark}{remarks}
\Crefname{remark}{Remark}{Remarks}

\newcommand{\C}{\mathbb C}
\newcommand{\R}{\mathbb R}
\newcommand{\Z}{\mathbb Z}

\newcommand{\Hcal}{\mathcal H}
\newcommand{\ket}[1]{\lvert #1\rangle}
\newcommand{\inner}[2]{\langle #1\mid #2\rangle}
\newcommand{\spann}{\operatorname{span}}
\newcommand{\one}{\mathbf 1}
\newcommand{\bic}{\operatorname{bc}}
\newcommand{\subjclass}[2][]{\gdef\papersubjclass{#2}}
\newcommand{\keywords}[1]{\gdef\paperkeywords{#1}}

\title{	Minimum Cardinalities of Multipartite Unextendible Product Bases}
\author{Chenhao Wang\\
 Beijing Normal University-Zhuhai ~\& ~BNBU\\chenhwang@bnu.edu.cn}
\date{}

\begin{document}

\maketitle

\begin{abstract}
In quantum information theory, the state space of a multipartite quantum
system is modeled by a tensor product.  In the tensor-product space
$\C^{d_1}\otimes\cdots\otimes\C^{d_p}$, a nonzero vector
is a \emph{product state} if it can be written as
$\ket{\varphi_1}\otimes\cdots\otimes\ket{\varphi_p}$ with
$\ket{\varphi_j}\in\C^{d_j}\setminus\{0\}$.  An \emph{unextendible product
basis} (UPB) is a finite family of pairwise orthogonal product states such
that no nonzero product state is orthogonal to all of them.
UPBs play a key role in investigating quantum
entanglement and nonlocal phenomena.  Finding a
smallest UPB is a natural extremal problem: it asks how few pairwise
orthogonal product states suffice to prevent any further product state from
being added. 
The general minimum-size problem for UPBs has been studied for over two decades since the seminal work of  Alon and Lov\'asz. 
 Let $f_m(d_1,\ldots,d_p)$ be the minimum
cardinality of a UPB in the space, and let
\(f_{LB}(d_1,\ldots,d_p)=1+\sum_{j=1}^{p}(d_j-1)\) be the natural lower bound.  Alon and Lov\'asz determined exactly when
\(f_m\) attains the lower bound \(f_{LB}\), but the obstructed multipartite cases remained open in
general. 

We prove a
stabilization theorem: for every non-all-qubit system with
$p\ge3$, whenever  the natural lower bound $f_{LB}$ is unattainable, the true minimum is exactly $f_{LB}+1$.  Precisely, if the number of even local
dimensions $d_i$ is positive and even, and at least one local dimension is greater
than two, then
\(f_m(d_1,\ldots,d_p)=f_{LB}(d_1,\ldots,d_p)+1\).
The proof is built on a unified graph-theoretic framework.
Our result, together with earlier work, settles the minimum size
problem for UPBs in all finite quantum systems.
\end{abstract}

% \newpage
% \tableofcontents

\section{Introduction}\label{sec:intro}

Consider the tensor-product space
\(
 \Hcal=\C^{d_1}\otimes\cdots\otimes\C^{d_p},
\)
where $\C^{d_j}$ is the $j$th local space.  A local space is
called a \emph{qubit} when $d_j=2$. A nonzero vector
$\ket{\varphi}\in\Hcal$ is a \emph{product state} if it admits a factorization
\[
 \ket{\varphi}
 =\ket{\varphi_1}\otimes\cdots\otimes\ket{\varphi_p},
 \qquad
 \ket{\varphi_j}\in\C^{d_j}\setminus\{0\}.
\]
A vector admitting no such factorization is called \emph{entangled}.  An
\emph{unextendible product basis} (UPB) is a finite family of pairwise
orthogonal product states such that no nonzero product state is orthogonal to
every member of the family.  Equivalently, it is an orthogonal family of
rank-one tensors that cannot be enlarged by another rank-one tensor while
preserving orthogonality.  The word ``basis'' is historical: a UPB generally
spans only a proper subspace of $\Hcal$, rather than forming a linear basis of
the whole space.  Its orthogonal complement contains no nonzero product
state, indicating that every nonzero vector in that complement is entangled.
Thus the definition combines the product-state requirement with orthogonality
and maximality conditions, naturally leading to questions in linear algebra
and combinatorics.

Originating in quantum information theory, UPBs have become
fundamental objects in the study of quantum entanglement and nonlocality.  Their orthogonal complements give standard examples of
bound entanglement, and the UPB states themselves illustrate ``nonlocality
without entanglement'': although each state is a product state and the states
are mutually orthogonal, spatially separated parties may be unable to
distinguish them perfectly using only local operations and classical
communication~\cite{BennettEtAl1999,DiVincenzoEtAl2003,DeRinaldis2004,Cohen2008}.
From a computational viewpoint, state discrimination asks how much
information a measurement procedure can recover from a quantum encoding.
UPBs give finite instances in which a global procedure succeeds, whereas a
distributed protocol restricted to local quantum operations and classical
communication can fail.  They therefore provide a useful setting for studying
the capabilities and limitations of distributed quantum protocols.  UPBs also
appear in the study of subspaces containing no product vectors, Bell
inequalities, and quantum networks~\cite{Parthasarathy2004,AugusiakEtAl2011,DuanXinYing2010}.
These applications motivate the problem, while the arguments below use only
graph theory and algebra.

The basic question is how few product states can form a UPB, since smaller
UPBs yield larger entangled complements and more economical constructions for
the applications above.  For a system
$\C^{d_1}\otimes\cdots\otimes\C^{d_p}$ with
local dimensions $d_1,\ldots,d_p\ge2$, let
$f_m(d_1,\ldots,d_p)$ denote the minimum cardinality of a UPB.  A simple
dimension-partition argument gives the natural lower bound
\[
 f_m(d_1,\ldots,d_p)\ge
 f_{LB}(d_1,\ldots,d_p)
 :=1+\sum_{j=1}^{p}(d_j-1).
\]
In their seminal work~\cite{AlonLovasz2001}, Alon and Lov\'asz initiated the
systematic study of minimum-cardinality UPBs and determined exactly when $f_m$
attains the lower bound $f_{LB}$.  They proved that this bound fails to be
attained if and only if one of the following two conditions holds: (i) the system is bipartite ($p=2$)
and one local dimension $d_j$ equals two; or (ii) $f_{LB}$ is odd and at least one local
dimension $d_j$ is even.

Determining the exact minimum $f_m$ in these obstructed cases has proved substantially more difficult.
  Chen and Johnston \cite{ChenJohnston2015} completed the bipartite classification and
settled a large range of multipartite systems in which one local dimension is
dominant.  Feng \cite{Feng2006} used decompositions into perfect
matchings to resolve several exceptional families, and
Johnston \cite{Johnston2013} determined the exact minimum for every all-qubit system.
For the system \(\mathbb C^2\otimes\mathbb C^2\otimes\mathbb C^d\) with odd $d$, Chen and
Johnston~\cite{ChenJohnston2015}  and Zhang et al.~\cite{ZhangEtAl2021} together proved
the minimum size of $d+3=f_{LB}+1$.
Some other special cases have also been shown to have a minimum size of $f_{LB}+1$, for example, the systems \((\mathbb C^2)^{\otimes t}\otimes\mathbb C^s\) with \(t,s\ge3, t+s+1\equiv0\pmod4\) \cite{ZhangEtAl2021}.
Recently, some new structural characterizations of the unextendibility
condition have been developed~\cite{ShiEtAl2023}, and automated procedures have
been introduced for constructing and verifying UPBs of prescribed
sizes~\cite{HanEtAl2026}.  Nevertheless, the exact minimum cardinality remained unknown for general
 multipartite systems in which \(f_{LB}\) is unattainable.

For $p\ge3$, the parity of the natural lower bound $f_{LB}$ is determined by
\[
 f_{LB}(d_1,\ldots,d_p)
 \equiv 1+|\{j:d_j\text{ is even}\}|\pmod2.
\]
Consequently, the Alon--Lov\'asz parity obstruction occurs exactly when the
number of even local dimensions is positive and even.  All-qubit systems ($d_j=2$ for all $j$) show
that the resulting excess above $f_{LB}$ can be larger than one; for example, $f_m(2,\ldots,2)=11$
 while $f_{LB}(2,\ldots,2)=9$ for $p=8$ qubits, resulting in an excess of $2$.  Our main result
shows that, among systems with \(p\ge3\), this larger excess is a purely all-qubit phenomenon: as soon as one
local dimension is greater than two, the parity obstruction costs exactly one
additional product state.  

\begin{samepage}
\begin{theorem}[Our main result]\label{thm:stabilization}
Let $p\ge3$ and $d_1,\ldots,d_p\ge2$.  Suppose that
$|\{j:d_j\text{ is even}\}|$ is positive and even, and that at least one local
dimension is greater than two.  Then 
\[
 \boxed{
 f_m(d_1,\ldots,d_p)
 =
 2+\sum_{j=1}^{p}(d_j-1)
 =
 f_{LB}(d_1,\ldots,d_p)+1.
 }
\]
\end{theorem}

\end{samepage}

The theorem extends the known obstructed multipartite results to all
non-all-qubit systems.  Together with the natural bound and the
known bipartite and all-qubit formulas,  it resolves the minimum
UPB size problem in every finite-dimensional system, yielding the following complete formula.

\begin{theorem}[Complete formula]\label{thm:classification}
For one party, $f_m(d)=d$.  For bipartite systems, assume
$2\le d_1\le d_2$.  Then
\begin{equation}\label{eq:bipartite-formula}
 f_m(d_1,d_2)=
 \begin{cases}
 d_1d_2,&d_1=2,\\
 d_1+d_2,&d_1,d_2\ge4\text{ are both even},\\
 d_1+d_2-1,&\text{otherwise}.
 \end{cases}
\end{equation}
For $p\ge3$, the all-qubit formula is
\begin{equation}\label{eq:all-qubit-formula}
 f_m(2,\ldots,2)=
 \begin{cases}
 p+1,&p\text{ is odd},\\
 p+2,&p=4\text{ or }p\equiv2\pmod4,\\
 11,&p=8,\\
 p+4,&p\ge12,\ p\equiv0\pmod4.
 \end{cases}
\end{equation}
For every multipartite system that is not all-qubit, put
$f_{LB}=1+\sum_j(d_j-1)$.  Then
\begingroup
\makeatletter
\let\UPBtagform\tagform@
\renewcommand{\tagform@}[1]{%
  \ifnum\value{equation}=4\relax
    \maketag@@@{\textbf{(#1)}}%
  \else
    \UPBtagform{#1}%
  \fi
}
\begin{numcases}{f_m(d_1,\ldots,d_p)=}
 f_{LB}, & if $|\{j:d_j\text{ is even}\}|=0$ or is odd,
 \label{eq:unobstructed-multipartite}\\
 f_{LB}+1, & if $|\{j:d_j\text{ is even}\}|$ is positive and even.
 \label{eq:obstructed-multipartite}
\end{numcases}
\endgroup
\end{theorem}

Our main contribution is the full generality of \Cref{thm:stabilization},
which gives \eqref{eq:obstructed-multipartite} beyond the previously known
families (e.g., $f_m(2,2,d)$ with odd $d$).   The one-party statement is immediate,
the bipartite formula  is due to Chen and
Johnston~\cite{ChenJohnston2015}, the unobstructed multipartite formula
\eqref{eq:unobstructed-multipartite} is due to Alon and
Lov\'asz~\cite{AlonLovasz2001}, and the all-qubit formula
is due to Johnston~\cite{Johnston2013}.

The complete formula also reveals a simple pattern in the excess
$f_m-f_{LB}$.  For the bipartite family $(2,d)$, the excess is $d-1$ and is
therefore unbounded.  Outside this family, the excess is at most three, and
values greater than one occur only for all-qubit systems.  
These larger
excesses are thus confined to special families.  In the general setting of
non-all-qubit multipartite systems, the minimum is always either $f_{LB}$ or
$f_{LB}+1$, and our contribution settles the previously open cases
by proving that their minimum is $f_{LB}+1$.

\paragraph{Technical overview} Our UPB construction methods build on a well-studied graph-theoretic framework as follows. In any $\C^{d_1}\otimes\cdots\otimes\C^{d_p}$,
 designate one local
system of dimension $d_*$ as
\emph{exceptional}, and all remaining local systems are called \emph{ordinary}.  In an
ordinary system of dimension $d_j$, we require every $d_j$ local vectors to
be linearly independent; equivalently, they span $\C^{d_j}$.  In the exceptional system, %a hyperplane may contain $d_*$ local vectors, one more than in an ordinary system.  
we require that every $d_*+1$ local vectors span $\C^{d_*}$.
For each local system, define its \emph{local orthogonality graph} with one vertex for
each product state, joining two vertices when the corresponding local vectors
in that system are orthogonal.  Since two
product states are orthogonal if and only if they are orthogonal in at least
one local system, the product states are pairwise orthogonal exactly when the
union of the local orthogonality graphs covers the complete graph.  This
graph approach goes back to Bennett et al.~\cite{BennettEtAl1999} and
was further developed by DiVincenzo et al.~\cite{DiVincenzoEtAl2003},
 Chen and Johnston~\cite{ChenJohnston2015}, and
Zhang et al.~\cite{ZhangEtAl2021}. 

For local vectors in the exceptional
system, we
introduce several new construction methods, for example, weighted Fourier orbits and projected incidence vectors, depending on the dimension range.  
Then the remaining edges
of the complete graph (those not in the exceptional local orthogonality graph) need to be split among ordinary local systems. 
The edges assigned to each ordinary subsystem of dimension $d_j$ form a graph.  We use additive-combinatorial estimates to show that
  its biclique number (i.e., the maximum total number of vertices in a complete bipartite subgraph) is at most $d_j$.
Then one can assign a local vector in $\R^{d_j}$ to each vertex so that adjacent vertices receive orthogonal vectors and every $d_j$ of the vectors
 are linearly independent, as required for the UPB construction.

The main difficulty is to make this framework work for every remaining tuple
of local dimensions.  Our key new perspective is to organize these cases not
by the number of even local dimensions, but by the parity of $q$, where
$2q=f_{LB}+1$.  When $q$ is even, we introduce a weighted Fourier-orbit
construction for the exceptional local vectors (\Cref{prop:even-library}) and
a uniform edge-allocation scheme based on a round-robin decomposition into
perfect matchings (\Cref{prop:even-factor}).  When $q$ is odd, we develop new
aligned Fourier constructions for the exceptional local vectors 
(\Cref{prop:odd-library}), and introduce an allocation method combining
aligned edge allocations with perfect-matching decompositions of Cayley graphs
(\Cref{prop:odd-allocation}).  Together with a separate argument for the
special low-dimensional cases, these ideas cover all remaining cases.

This paper is organized as follows.  In \Cref{sec:model}, we introduce the
basic definitions and notation.  In \Cref{sec:framework}, we develop the
graph-theoretic framework for our main constructions.
\Cref{sec:even,sec:odd} apply this framework to the cases of even and odd
$q$, respectively.
\Cref{sec:outlook} discusses further questions. Appendix~\ref{app:cyclic} gives further biclique bounds for cyclic annuli and two-layer graphs. 
Appendices~\ref{app:odd-exceptional} and~\ref{app:odd-allocation} contain
the proofs of \Cref{prop:odd-library,prop:odd-allocation}, respectively.
 Appendix~\ref{app:theta} gives an independent proof of the
known special case $\C^{2}\otimes\C^{2}\otimes\C^{d}$ with odd $d$.

%%%%%%%%%%%%%%%%%%%%%%%%%%%%%%%%%%%%%%%%%%%%%%%%%
\section{Preliminaries}\label{sec:model}

A nonzero vector $\ket{\varphi}\in\bigotimes_{j=1}^{p}\C^{d_j}$ is a
\emph{product state} if it can be written as
\(
 \ket{\varphi}
 =
 \ket{\varphi_1}\otimes\cdots\otimes\ket{\varphi_p},
\) where $\ket{\varphi_j}\in\C^{d_j}\setminus\{0\}$.
Otherwise, it is called an \emph{entangled state}.  
% Let
% \(
%  \ket{0}=(1,0)^{\mathsf T}
% \)
% and
% \(
%  \ket{1}=(0,1)^{\mathsf T}
% \)
% denote the standard basis vectors of $\C^2$.  
For example,
$\ket{0}\otimes\ket{1}$ is a product state in $\C^2\otimes\C^2$, whereas
\(
 \ket{0}\otimes\ket{0}+\ket{1}\otimes\ket{1}
\)
is entangled.  Indeed, if it were equal to
$(a\ket{0}+b\ket{1})\otimes(c\ket{0}+d\ket{1})$, comparison of coefficients
would give $ac=bd=1$ and $ad=bc=0$, which is impossible. The system $\bigotimes_{j=1}^{p}\C^{d_j}$ 
is \emph{bipartite} if $p=2$ and \emph{multipartite} if $p\ge3$.
A local system $\C^{d_j}$ is called a \emph{qubit} if its dimension $d_j$ is $2$ and a
\emph{qutrit} if $d_j=3$. 
Since local systems of dimension $1$ are trivial and permuting the tensor factors does not affect the problem, throughout
this paper we assume without loss of generality that
\(
 2\le d_1\le\cdots\le d_p.
\)

\begin{definition}[Unextendible product basis, UPB]\label{def:upb}
A finite family of product states is an \emph{unextendible product basis} if its members are pairwise orthogonal and no nonzero product state is orthogonal to all of them.
\end{definition}

The minimum cardinality of a UPB in $\bigotimes_{j=1}^{p}\C^{d_j}$ is denoted by 
$f_m(d_1,\ldots,d_p)$.  Alon and Lov\'asz~\cite{AlonLovasz2001} gave a natural lower bound 
for this quantity:
\[
 f_m(d_1,\ldots,d_p)\ge
 f_{LB}(d_1,\ldots,d_p)
 :=
 1+\sum_{j=1}^{p}(d_j-1).
\] 
Indeed, if a family contains at most $\sum_j(d_j-1)$ states, partition its indices 
into $p$ parts of sizes at most $d_j-1$, and  choose a nonzero vector orthogonal 
to the local vectors indexed by the $j$th part for each $j$; then the tensor product of these choices would be orthogonal to the whole family.
We simply write this lower bound as $f_{LB}$ when the context is clear, and write $$2q=f_{LB}+1=2+\sum_j(d_j-1),\qquad \text{for odd } f_{LB}.$$
Write $[n]=\{1,2,\ldots,n\}$ for every positive integer $n$.

Given $n$ product states
\(
 \ket{\psi_i}
 =
 \ket{v_i^{(1)}}\otimes\cdots\otimes\ket{v_i^{(p)}}
\), $i\in[n]$,
the \emph{local orthogonality graph} $G_j$ in local dimension $j\in[p]$ has vertex set $[n]$ and edge $ih$ when
\(
 v_i^{(j)}\perp v_h^{(j)}.
\)
The states are pairwise orthogonal exactly when their union is the complete graph $K_n$:
\[
 G_1\cup\cdots\cup G_p=K_n.
\]

We recall the known results in the literature.

\begin{theorem}[Alon--Lov\'asz~\cite{AlonLovasz2001}]\label{thm:AL}
The equality
\(
 f_m(d_1,\ldots,d_p)=f_{LB}(d_1,\ldots,d_p)
\)
holds unless one of the following two cases occurs:
\begin{enumerate}
\item $p=2$ and one local dimension is two;
\item $f_{LB}$ is odd and at least one local dimension $d_j$ is even.
\end{enumerate}
In either case, the minimum is strictly greater than $f_{LB}$.
\end{theorem}

\begin{theorem}[Chen--Johnston~\cite{ChenJohnston2015}]\label{thm:CJ}
The bipartite formula for $p=2$ is given in \eqref{eq:bipartite-formula}.
For multipartite systems in which the Alon--Lov\'asz lower bound is not
attainable, if $d_p\ge q$ and
\(
 \sum_{j<p}(d_j-1)\ge3,
\)
then
\(
 f_m(d_1,\ldots,d_p)=f_{LB}(d_1,\ldots,d_p)+1.
\)
\end{theorem}

Note that $d_p\ge q$ means that \(d_p\) is a dominant local dimension: $d_p-1$ is at least the sum of all other $d_j-1$.
The only multipartite case in which $d_p\ge q$ but
\(
\sum_{j<p}(d_j-1)<3
\)
is the tripartite family $(2,2,d_p)$, because
$d_j-1\ge1$ for every local dimension and therefore
\(
\sum_{j<p}(d_j-1)<3
\)
forces exactly two remaining parties, both of dimension two.
For $(2,2,d)$ with even $d$, the lower bound $f_{LB}=d+2$ is even,
so \Cref{thm:AL} already gives
\(
 f_m(2,2,d)=f_{LB}.
\)
For odd $d$, the following formula is also known.

\begin{theorem}[$(2,2,d)$ with odd $d$]\label{thm:22d-known}
For every odd integer $d\ge3$,
\(
 f_m(2,2,d)=f_{LB}+1=d+3.
\)
This was proved by Chen--Johnston~\cite{ChenJohnston2015}
for $d\equiv1\pmod4$ and by Zhang et al.~\cite{ZhangEtAl2021}
for $d\equiv3\pmod4$.
\end{theorem}

In Appendix~\ref{app:theta},  we independently give a unified explicit construction of minimum UPBs in \(\mathbb C^2\otimes\mathbb C^2\otimes\mathbb C^d\) for every odd \(d\ge3\), 
using projected incidence vectors of a theta graph.

Consequently, \Cref{thm:AL,thm:CJ,thm:22d-known} settle all cases with
$d_p\ge q$.  To prove our \Cref{thm:stabilization}, it therefore suffices to
establish the upper bound $f_{LB}+1$ for the remaining obstructed
non-all-qubit systems with $d_p<q$.

%%%%%%%%%%%%%%%%%%%%%%%%%%%%%%%%%%%%%%%%%%%%%%
\section{A Graph-Theoretic Framework}\label{sec:framework}
In this section, 
as a warm-up we begin with the lower bound of $f_{LB}+1$ for remaining cases, 
and then provide a general framework for the upper bound proof.

The following is a direct edge-counting reformulation of the
Alon--Lov\'asz parity obstruction, and we write the proof here merely for completeness.

\begin{proposition}[Lower bound \cite{AlonLovasz2001}]\label{prop:lower}
Let $p\ge3$.  If the number of even local dimensions is positive and even, then
\[
 f_m(d_1,\ldots,d_p)
 \ge
 2+\sum_{j=1}^{p}(d_j-1)
 =
 f_{LB}+1.
\]
\end{proposition}

\begin{proof}
The number $f_{LB}$ is odd.  Suppose, for a contradiction, that there is a size-$f_{LB}$ UPB
\[
 \psi_i=v_i^{(1)}\otimes\cdots\otimes v_i^{(p)},
 \qquad i\in[f_{LB}].
\]
We first show that, in every subsystem $j$, any $d_j$ of the local vectors $v_i^{(j)}$ are linearly independent.  
Otherwise, there would be a set $I_j\subseteq[f_{LB}]$ of size $d_j$ such that the vectors $v_i^{(j)}$, $i\in I_j$, span a proper subspace of $\C^{d_j}$.  
We could then choose a nonzero vector $w_j\in\C^{d_j}$ orthogonal to all of them.
The number of remaining indices is
\[
 f_{LB}-d_j
 =\sum_{h\ne j}(d_h-1).
\]
Partition them into sets $I_h$, $h\ne j$, with $|I_h|=d_h-1$.  For every $h\ne j$, the vectors $v_i^{(h)}$, $i\in I_h$, span a proper subspace of $\C^{d_h}$, so choose a nonzero $w_h\in\C^{d_h}$ orthogonal to all of them.  The product vector
\(
 w_1\otimes\cdots\otimes w_p
\)
is orthogonal to $\psi_i$ through subsystem $j$ when $i\in I_j$, and through subsystem $h$ when $i\in I_h$.  It is therefore orthogonal to the entire family, contradicting unextendibility.

Now consider the local orthogonality graph $G_j$.  If a vertex $i$ had at least $d_j$ neighbors, the corresponding $d_j$ local vectors would all lie in the hyperplane $(v_i^{(j)})^\perp$ and hence would be linearly dependent.  Thus
the maximum degree of $G_j$ is at most $d_j-1$.
If $d_j$ is even, both $f_{LB}$ and $d_j-1$ are odd.  Since the sum of the vertex degrees is even, the degree bound gives
\[
 |E(G_j)|\le\frac{f_{LB}(d_j-1)-1}{2}.
\]
If $d_j$ is odd, it gives
\[
 |E(G_j)|\le\frac{f_{LB}(d_j-1)}2.
\]
Let $e>0$ be the number of even local dimensions.  Summing these bounds over all subsystems yields
\[
 \sum_j|E(G_j)|
 \le
 \frac{f_{LB}\sum_j(d_j-1)-e}{2}
 =
 \binom{f_{LB}}{2}-\frac e2
 <
 \binom{f_{LB}}{2}.
\]
Hence the local graphs contain, even when counted with multiplicity, fewer edges than $K_{f_{LB}}$.  They cannot cover $K_{f_{LB}}$, contradicting the pairwise orthogonality of the UPB.
\end{proof}

The remaining results establish a unified framework for the upper bound. Recall that $2q=f_{LB}+1=2+\sum_j(d_j-1)$. The next key lemma gives a sufficient condition for a family of
product states to form a UPB, and serves as the organizing
principle for every new construction: a family of
$2q$ pairwise orthogonal product states forms a UPB whenever
every $d_t+1$ local vectors span $\C^{d_t}$ in one subsystem $t$, while
every $d_j$ local vectors span $\C^{d_j}$ in each other
subsystem $j\ne t$.
We do not claim originality for this lemma. The same counting argument is used in the proof of
Theorem~1 of Chen and Johnston~\cite{ChenJohnston2015}.
Zhang et al.~\cite{ZhangEtAl2021} further state a general
criterion: an orthogonal family of product states is
unextendible if and only if its vertex set cannot be covered
by choosing one unsaturated set from each local subsystem.
Here, a set is unsaturated in a subsystem if its corresponding
local vectors span a proper subspace.
Our lemma gives a sufficient condition for this criterion
by bounding the cardinalities of these unsaturated sets.

%a related counting argument can be found in the work of Feng~\cite{Feng2006}.

\begin{lemma}\label{lem:one-extra}
Assume that nonzero local vectors $v_i^{(j)}\in\C^{d_j}$ ($1\le i\le 2q$, $1\le j\le p$) have local orthogonality graphs covering $K_{2q}$.  Suppose that, for one subsystem $t$, every $d_t+1$ vectors among $v_1^{(t)},\ldots,v_{2q}^{(t)}$ span $\C^{d_t}$, while for every $j\ne t$, every $d_j$ vectors among $v_1^{(j)},\ldots,v_{2q}^{(j)}$ are linearly independent.  Then the associated $2q$ product states form a UPB.
\end{lemma}

\begin{proof}
The graph cover gives pairwise orthogonality.  A hyperplane in an ordinary subsystem $j\neq t$ contains at most $d_j-1$ local vectors, 
while a hyperplane in the exceptional subsystem $t$ contains at most $d_t$ local vectors.  A nonzero product state can therefore be orthogonal to at most
\[
 d_t+\sum_{j\ne t}(d_j-1)=2q-1
\]
of the $2q$ states.
\end{proof}

Therefore, to construct a UPB of size $2q$, it suffices
to find nonzero local vectors $v_i^{(j)}$ such that (1) their local orthogonality graphs
cover $K_{2q}$, (2) every $d_t+1$ vectors in one subsystem $t$ span
$\C^{d_t}$, and (3) every $d_j$ vectors in each remaining subsystem $j\ne t$
are linearly independent.  We call $t$ the \emph{exceptional subsystem} and
the remaining subsystems the \emph{ordinary subsystems}. 

One of the main challenges is to construct local vectors in
the ordinary subsystems so that every $d_j$ vectors in
subsystem $j$ are linearly independent and the union of
their local orthogonality graphs covers the complement of
the exceptional orthogonality graph.
We reduce this task to finding a graph
cover with suitable bounds on the biclique numbers.

An \emph{orthogonal representation} of a graph \(G\) in \(\mathbb R^d\) is an assignment of nonzero vectors
\(
z_v\in\mathbb R^d
\) to every vertex \(v\in V(G)\), such that
\[
uv\in E(G)\quad\Longrightarrow\quad z_u\perp z_v.\]
Define
 $\bic(G)$, the \emph{biclique number}, to be the largest number of vertices in a complete bipartite subgraph of $G$.  The subgraph need not be induced.  We set $\bic(G)=0$ if $G$ has no edges.

\begin{lemma}\label{lem:bic-connect}
Let $1\le d<N$.  If $G$ has $N$ vertices and $\bic(G)\le d$, then
$G$ has an orthogonal representation in $\R^d$ in which every $d$ vectors
are linearly independent.
\end{lemma}

\begin{proof}
A graph is \emph{$k$-connected} if it has more than $k$ vertices and remains
connected after the deletion of any set of fewer than $k$ vertices.
If deleting at most $N-d-1$ vertices disconnected $\overline G$, two nonempty unions of the remaining components would have every cross pair in $G$.  They would form a complete bipartite subgraph on at least $d+1$ vertices, contrary to $\bic(G)\le d$.  Thus $\overline G$ is $(N-d)$-connected.

The theorem of Lov\'asz, Saks, and Schrijver~\cite{LovaszSaksSchrijver1989,LovaszSaksSchrijver2000} states that a graph $F$ on $N$ vertices is $k$-connected if and only if there exist real vectors $z_v\in\R^{N-k}$, $v\in V(F)$, such that nonadjacent vertices receive orthogonal vectors and every $N-k$ vectors are linearly independent.  Apply this theorem to $F=\overline G$ with $k=N-d$.  It gives vectors in $\R^d$ such that every $d$ are linearly independent.  Two vertices are nonadjacent in $\overline G$ precisely when they are adjacent in $G$, so these vectors form the required orthogonal representation of $G$.
\end{proof}

% While the dominant-$d_p$ cases with $d_p\ge q$ is
% already settled by \Cref{thm:AL,thm:CJ,thm:22d-known}, i
In the next two sections, we establish the upper bound for the systems with
$d_p<q$, according as $q$ is even or odd.  
Both \Cref{sec:even,sec:odd} use the framework developed
in this section: they construct exceptional local vectors satisfying
\Cref{lem:one-extra}, partition the remaining edges into graphs for the
ordinary subsystems, and apply \Cref{lem:bic-connect} to obtain the ordinary
local vectors.

%%%%%%%%%%%%%%%%%%%%%%%%%%%%%%%%%%%%%%%%%%%%%%%%%%%%%%%%%%
\section{Even-\texorpdfstring{$q$}{q} Case}\label{sec:even}

This section completes the proof of the upper bound for even $q$, where $2q=2+\sum_j(d_j-1)=f_{LB}+1$.
The dominant-$d_p$ cases with \(d_p\ge q\) have already been handled, so the new
construction below assumes \(d_p<q\). Since $q=2$ reduces to the bipartite case $(2,2)$, we assume $q\ge 4$. 

We choose an index $t$ for which $d_t$ is even, designate
subsystem $t$ as the exceptional subsystem in \Cref{lem:one-extra}, and write
$d_*=d_t$.  Thus its local vectors are required to have the exceptional
spanning property that every $d_*+1$ of them span $\C^{d_*}$.
Index its
$2q$ local vectors by two sets
\[
 U=\{u_i:i\in\Z_q\},
 \qquad
 V=\{v_i:i\in\Z_q\},
\]
which we call the two \emph{layers}. Here $\Z_q=\Z/q\Z$ denotes the additive cyclic group of
integers modulo $q$, represented by $\{0,1,\ldots,q-1\}$.
Let $G_*$ denote the orthogonality graph of these exceptional local
vectors on $U\cup V$: two vertices are adjacent in $G_*$ exactly when
their corresponding vectors are orthogonal.  We call $G_*$ the
\emph{exceptional orthogonality graph}.
For $s\in\Z_q$, the \emph{cross-layer matching of shift $s$} is
\[
 M_s=\{u_i v_{i+s}:i\in\Z_q\}.
\]
We say that $G_*$ has \emph{cross-layer degree} $r$ if every vertex has
exactly $r$ neighbors in the opposite layer.  The
matchings $M_s$, $s\in\Z_q$, partition all
edges between the two layers.  We call $M_s$ \emph{unused} if none of its
edges belongs to $G_*$; it then lies in $\overline{G_*}$ and is available
for an ordinary subsystem.  A \emph{cyclic interval}
of $L$ unused cross-layer matchings means
$M_a,M_{a+1},\ldots,M_{a+L-1}$ for some $a\in\Z_q$, with shifts taken modulo
$q$.
When $q$ is even, the \emph{antipodal matchings} in the two layers are
\[
 \{u_i u_{i+q/2}:0\le i<q/2\}
 \quad\text{and}\quad
 \{v_i v_{i+q/2}:0\le i<q/2\},
\]
respectively.

The proof proceeds as follows.  \Cref{prop:even-library} constructs the
$2q$ exceptional local vectors, every $d_*+1$ of which span $\C^{d_*}$,
and determines their orthogonality graph $G_*$.  Then,
\Cref{prop:even-factor} partitions $\overline{G_*}$ into $p-1$ spanning
graphs $G_i$, one for each ordinary subsystem, such that
$\bic(G_i)\le d_i$.  By \Cref{lem:bic-connect}, each $G_i$ has an
orthogonal representation in $\R^{d_i}$ in which every $d_i$ vectors are
linearly independent.  Since $G_*$ together with the ordinary graphs $G_i$
partitions $K_{2q}$, every two of the resulting product states are orthogonal.  Finally, \Cref{lem:one-extra} shows that these $2q$
product states form a UPB.

\begin{proposition}\label{prop:even-library}
Let $q\ge4$ be even.
\begin{enumerate}[label=(\arabic*)]
\item There are $2q$ vectors in $\C^2$ such that (i) every three span, and
(ii) the graph induced by each layer is exactly its antipodal matching,
while the cross-layer edges form exactly the shift-zero matching $M_0$.

\item For every $4\le d\le q-1$, there are $2q$ vectors in $\C^d$ such
that (i) every $d+1$ span, and (ii) the graph induced by each layer is
exactly its antipodal matching, while the full graph has cross-layer degree
$d-1$.  The unused cross-layer matchings form a cyclic interval of length
$q-d+1$.
\end{enumerate}
\end{proposition}

\begin{proof}

We prove both parts by explicitly constructing the required 
$2q$ local vectors.

\paragraph{Part (1).}

Assume that $q$ is even.  For $0\le i<q/2$, choose nonzero vectors $x_i,y_i\in\C^2$ with $x_i\perp y_i$ such that the one-dimensional subspaces
$\C x_i$ and $\C y_i$ are all distinct.  Define
\[
 c_{u_i}=c_{v_{i+q/2}}=x_i,
 \qquad
 c_{u_{i+q/2}}=c_{v_i}=y_i.
\]
The prescribed graph is a disjoint union of four-cycles.  It contains the antipodal matching in each layer and the cross-layer matching $u_jv_j$.  Every one-dimensional subspace occurs exactly twice, so no three vectors lie in a common one-dimensional subspace.  Hence every three span $\C^2$.

\paragraph{Part (2).}

Assume that $q$ is even and $4\le d\le q-1$.  Put $r=q-d$ and let
$\zeta=e^{2\pi i/q}$, a \emph{primitive $q$th root of unity}, meaning that
$1,\zeta,\ldots,\zeta^{q-1}$ are all distinct, and define
\begin{equation}\label{eq:even-P}
 P(z)=\prod_{h=d}^{q-1}(z-\zeta^h).
\end{equation}
The polynomial has degree $r$.  If $\alpha=\zeta^d$, the Gaussian-binomial expansion is
\begin{equation}\label{eq:q-binomial}
 \prod_{t=0}^{r-1}(z-\alpha\zeta^t)
 =\sum_{k=0}^{r}(-1)^k\alpha^k\zeta^{k(k-1)/2}
 \begin{bmatrix}r\\k\end{bmatrix}_{\!\zeta}z^{r-k},
\end{equation}
where the bracket denotes the Gaussian binomial coefficient
\(
 \begin{bmatrix}r\\k\end{bmatrix}_{\!\zeta}
 =\prod_{j=1}^{k}\frac{1-\zeta^{r-k+j}}{1-\zeta^j}.
\)
Every exponent in the numerator and denominator lies between $1$ and $r<q$.  Thus all factors, and hence all coefficients of $P$, are nonzero.

For $0\le j<d$, the values $P(\zeta^j)$ are nonzero.  For a small
positive parameter $\varepsilon$, define
\[
\begin{aligned}
 \lambda_0&=1+x,\qquad \lambda_1=1,\qquad \lambda_j=\varepsilon^{j-1}\quad(2\le j\le d-3),\\
 \lambda_{d-2}&=\varepsilon^{d-3}(|P(\zeta^{d-2})|^2+y),
 \qquad\lambda_{d-1}=\varepsilon^{d-3}|P(\zeta^{d-1})|^2.
\end{aligned}
\] We choose the real variables $x,y$ so that
\begin{equation}\label{eq:even-balance-general}
 \sum_{j=0}^{d-1}(-1)^j\lambda_j=0,
 \qquad
 \sum_{j=0}^{d-1}(-1)^j\frac{|P(\zeta^j)|^2}{\lambda_j}=0.
\end{equation}
After substitution, multiply the second equation by $\varepsilon^{d-3}$ and
cancel the explicit powers of $\varepsilon$.  The resulting pair of real
equations extends smoothly to $\varepsilon=0$ near $(x,y)=(0,0)$.  At that
point the first equation is $1-1=0$, and the scaled second equation is the
cancellation of its last two terms.  The Jacobian matrix, namely the matrix
of first derivatives with respect to $(x,y)$, is diagonal and invertible.
The implicit function theorem says that, near such a solution with
invertible Jacobian, the equations can be solved smoothly for $(x,y)$ as
functions of the remaining parameter.  It therefore gives solutions
$x(\varepsilon),y(\varepsilon)\to0$.  All weights are positive for
sufficiently small $\varepsilon$.

For $i\in\Z_q$, define vectors
\begin{equation}\label{eq:even-fourier-vectors}
 c_{u_i}=\bigl(\sqrt{\lambda_j}\,\zeta^{ji}\bigr)_{j=0}^{d-1},\qquad
 c_{v_i}=\left(\frac{P(\zeta^j)}{\sqrt{\lambda_j}}\,\zeta^{ji}\right)_{j=0}^{d-1}.
\end{equation}
These two families are \emph{weighted Fourier orbits}: increasing the index
$i$ multiplies coordinate $j$ by $\zeta^j$.  For $i,\ell\in\Z_q$,
put $z=\zeta^{\ell-i}$.  The within-layer inner products are
\[
 \inner{c_{u_i}}{c_{u_\ell}}=\sum_{j=0}^{d-1}\lambda_jz^j,
 \qquad
 \inner{c_{v_i}}{c_{v_\ell}}=\sum_{j=0}^{d-1}\frac{|P(\zeta^j)|^2}{\lambda_j}z^j.
\]
Since $\zeta^{jq/2}=(-1)^j$, the two equations in
\eqref{eq:even-balance-general} give $c_{u_i}\perp c_{u_{i+q/2}}$ and
$c_{v_i}\perp c_{v_{i+q/2}}$.  These are the only layer orthogonalities
when $\varepsilon$ is sufficiently small.  Indeed, as $\varepsilon\to0$,
\[
 \inner{c_{u_i}}{c_{u_\ell}}\longrightarrow1+z,
 \qquad
 \varepsilon^{d-3}\inner{c_{v_i}}{c_{v_\ell}}
 \longrightarrow z^{d-2}(1+z).
\]
Since $z=\zeta^{\ell-i}$ is a $q$th root of unity, both displayed limits
vanish only when $z=-1$, equivalently $\ell=i+q/2$.  In this case,
\eqref{eq:even-balance-general} gives
$\inner{c_{u_i}}{c_{u_\ell}}=\inner{c_{v_i}}{c_{v_\ell}}=0$ exactly.
For each of the finitely many other values $z=\zeta^{\ell-i}\ne-1$,
both displayed limits are nonzero.  Hence
$\inner{c_{u_i}}{c_{u_\ell}}$ and $\inner{c_{v_i}}{c_{v_\ell}}$
are nonzero for sufficiently small $\varepsilon$.

For a cross-layer pair,
\begin{equation}\label{eq:even-cross-transform}
\begin{aligned}
 \inner{c_{u_i}}{c_{v_\ell}}
 &=\sum_{j=0}^{d-1}P(\zeta^j)\zeta^{j(\ell-i)}=\sum_{j=0}^{q-1}P(\zeta^j)\zeta^{j(\ell-i)}.
\end{aligned}
\end{equation}
The second equality uses the roots in \eqref{eq:even-P}.  To determine
directly which shifts give nonzero values, write $P(z)=\sum_k a_kz^k$ and
sum over the $q$th roots of unity;
the identity $\sum_{j=0}^{q-1}\zeta^{jm}=0$ unless $m\equiv0\pmod q$
selects one coefficient $a_k$.  This is the finite Fourier-inversion step.
Since every coefficient of $P$ is nonzero,
\eqref{eq:even-cross-transform} is nonzero exactly when
$\ell-i\in\{0,-1,\ldots,-r\}$.
Thus the prescribed exceptional orthogonality graph has one antipodal
matching in each layer and cross-layer degree $q-(r+1)=d-1$. Its unused
cross-layer matchings form the stated interval of length $r+1=q-d+1$.

It remains to verify the spanning property.  Let $F_{j,i}=\zeta^{ji}$ for
$0\le j<d$ and $i\in\Z_q$.  A \emph{Vandermonde matrix} has columns
$(1,z,z^2,\ldots)^T$ at distinct nodes $z$; its determinant is
$\prod_{i<j}(z_j-z_i)$ and is therefore nonzero.  Every square submatrix of
$F$ whose rows are consecutive is a nonzero diagonal multiple of such a
matrix.  The determinant of a square submatrix is called a \emph{minor}.

Choose $s$ vectors $c_{u_i}$ from the $U$ layer, with index set $I$, and
$d-s$ vectors $c_{v_j}$ from the $V$ layer, with index set $J$. Place these
$d$ vectors as the columns of a $d\times d$ matrix.  For a row set $K$, let
$F_{K,I}$ denote
the submatrix of $F$ on rows $K$ and columns $I$, and let $K^c$ denote the
complementary row set.  After multiplying row $j$ by
$\sqrt{\lambda_j}/P(\zeta^j)$, apply the Laplace expansion to the $d\times d$
matrix along the $s$ columns occupied by the selected $U$-layer vectors.  The
determinant is
\begin{equation}\label{eq:mixed-det-general}
 \sum_{\substack{K\subseteq\{0,\ldots,d-1\}:|K|=s}}
 \epsilon_K\det F_{K,I}\det F_{K^c,J}
  \prod_{k\in K}\frac{\lambda_k}{P(\zeta^k)}.
\end{equation}
Here $\epsilon_K\in\{1,-1\}$ is the sign arising from the determinant
expansion.  For each $k$, the order of the factor
$\lambda_k/P(\zeta^k)$ is the exponent of its leading power of
$\varepsilon$.  Since $P(\zeta^k)$ is nonzero and does not depend on
$\varepsilon$, these orders are the same as those of the $\lambda_k$:
\(
 0,0,1,2,\ldots,d-4,d-3,d-3.
\)
The order of the product in the term indexed by $K$ is the sum of the orders
of the $s$ factors $\lambda_k/P(\zeta^k)$ with indices in $K$.  Thus, for
$2\le s\le d-2$, this sum is uniquely minimized by choosing the $s$ smallest
orders, namely $K=\{0,1,\ldots,s-1\}$.  Hence the corresponding term is the
unique term of minimum order in \eqref{eq:mixed-det-general}.
Both Fourier minors in that term have consecutive rows and are nonzero.  Hence every such mixed $d$-tuple is independent for sufficiently small $\varepsilon$.  Pure $d$-tuples in either layer are Vandermonde families.  Now take any $d+1$ vectors and let $a$ be the number chosen from the $U$ layer.  If $a\le1$, retain the $d$ vectors from the $V$ layer; if $a\ge d$, retain the $d$ vectors from the $U$ layer.  Otherwise, delete one vector so that the retained $d$-tuple contains between $2$ and $d-2$ vectors from $U$.  The retained tuple is independent by the preceding arguments.  Thus every $d+1$ vectors contain an independent $d$-tuple and hence span $\C^d$, proving Part (2) of \Cref{prop:even-library}.
\end{proof}

For the chosen exceptional subsystem, construct its $2q$ local vectors
and its orthogonality graph $G_*$ using 
\Cref{prop:even-library}.
Then the complement $\overline{G_*}$  has
two parts:
\begin{enumerate}
\item across the layers, \(q-d_*+1\) consecutive unused cross-layer
perfect matchings;
\item inside each layer, the complete graph $K_q$ with one perfect matching removed. 
\end{enumerate}
For each copy of $K_q$, choose a round-robin decomposition into \(q-1\) perfect
matchings, one of which is the antipodal matching, and remove that
matching because it is contained in $G_*$.  Pair the corresponding remaining
matchings in the two layers.  This gives $q-2$ perfect matchings on the full
$2q$-vertex set.
%The cross-layer part of $\overline{G_*}$ already consists of $q-d_*+1$ perfect matchings.

The next proposition partitions \(\overline{G_*}\) by grouping these perfect matchings,
where each group corresponds to an ordinary subsystem.

\begin{proposition}\label{prop:even-factor}
Let $q\ge6$ be even, let $2\le d_*\le q-2$ be even, and let $d_1,\ldots,d_m$ be the remaining ordinary local dimensions.  Assume
that \(
 2\le d_i\le q-1
\) for all $1\le i\le m$,
 that the number of even \(d_i\)'s is odd, and that
\(
 \sum_{i=1}^{m}(d_i-1)=(q-2)+(q-d_*+1).
\)
Then the $q-2$ paired perfect matchings and the $q-d_*+1$ consecutive cross-layer matchings of $\overline{G_*}$ can be 
partitioned into spanning subgraphs $G_1,\ldots,G_m$ such that each \(G_i\) satisfies
\(
 \bic(G_i)\le d_i.
\)
\end{proposition}

We prove the proposition by a sequence of lemmas in additive combinatorics as follows. 

A round-robin decomposition of
a complete graph is a partition of all its edges into perfect matchings.  The
main graph estimate concerns consecutive perfect matchings in the following
standard decomposition.
Let $q\ge6$ be even, and label the vertices of the complete
graph $K_q$ by
$\{\infty\}\cup\Z_{q-1}$, where $\infty$ is one additional formal symbol.  For
$t\in\Z_{q-1}$, define
\[
 F_t=
 \{\{\infty,t\}\}
 \cup
 \bigl\{\{t+x,t-x\}:1\le x\le(q-2)/2\bigr\}.
\]
Because $q-1$ is odd, the nonzero elements of $\Z_{q-1}$ are partitioned
into the pairs $\{x,-x\}$ for $1\le x\le(q-2)/2$.  Thus the edge
$\{\infty,t\}$ covers $\infty$ and $t$, while the remaining edges cover
each other element of $\Z_{q-1}$ exactly once.  Hence $F_t$ is a perfect
matching.  As $t$ varies, the perfect matchings $F_t$ partition the edge set
of $K_q$ and hence form a round-robin decomposition.
In each layer, choose these auxiliary vertex labels so that $F_0$ is the
removed antipodal matching.  The available layer matchings are then
$F_1,\ldots,F_{q-2}$. The original cyclic labels are retained when
describing cross-layer matchings.

For subsets
$X,Y\subseteq\Z_m$, their \emph{sumset} and
\emph{difference set} are $X+Y=\{x+y:x\in X,\ y\in Y\}$ and
$Y-X=\{y-x:x\in X,\ y\in Y\}$, respectively.  For a set $E$ in a finite
abelian group, its \emph{stabilizer} is the subgroup
$\{h:E+h=E\}$.  If this subgroup is $H$, an \emph{$H$-coset} is a translate
$a+H$, and $E$ is a union of $H$-cosets.  A set is \emph{$H$-periodic} if
it is a union of $H$-cosets.  We call $E$ \emph{periodic} if its stabilizer
is nontrivial (i.e., contains more than one element) and \emph{aperiodic} otherwise.
We use Kneser's addition theorem~\cite{Kneser1953} in the following form:
if $X,Y$ are nonempty subsets of a finite abelian group and $H$ is the
stabilizer of $X+Y$, then
\begin{equation}\label{eq:kneser}
 |X+Y|\ge |X+H|+|Y+H|-|H|.
\end{equation}

A \emph{cyclic interval} of length $\ell$ in $\Z_m$ is a set of the form
$\{a,a+1,\ldots,a+\ell-1\}$, with all entries reduced modulo $m$; thus the
interval may wrap around from $m-1$ to $0$.  The following bound applies
both to blocks of cross-layer matchings and to the round-robin blocks below.

\begin{lemma}\label{lem:interval}
Let $m\ge2$ and let $X,Y\subseteq\Z_m$ be nonempty.  If $Y-X\subseteq I$, where $I$ is a
cyclic interval of size $\ell<m$, then $|X|+|Y|\le\ell+1$.
\end{lemma}

\begin{proof}
Let $E=Y-X$, and let $H$ be its stabilizer.  Write $|H|=h$ and $m=hg$.
Applying Kneser's theorem to $-X$ and $Y$, and using
$|-X+H|\ge |X|$ and $|Y+H|\ge |Y|$, gives
\begin{equation}\label{eq:kneser-basic}
 |X|+|Y|\le |E|+h.
\end{equation}
If $h=1$, the result follows from $|E|\le\ell$.

Assume $h>1$.  The complementary cyclic interval has
$m-\ell$ elements.  If $m-\ell\ge g$, it meets every $H$-coset, so no
nonempty $H$-periodic set can lie inside $I$.  Hence $m-\ell<g$.
Its $m-\ell$ consecutive points lie in distinct $H$-cosets, so $E$ uses
at most $g-(m-\ell)$ cosets.  Therefore
\[
 |E|+h\le h\bigl(g-(m-\ell)\bigr)+h
 =m-h(m-\ell-1)\le\ell+1,
\]
which gives $|X|+|Y|\le\ell+1$.
\end{proof}

We use the following interval form of the equality case in Kemperman's structure theorem~\cite{Kemperman1960}.  %It is included explicitly because only this short corollary is needed below.

\begin{lemma}[Kemperman~\cite{Kemperman1960} Theorem~5.2, Lev~\cite{Lev2006} Theorem~C]\label{lem:critical-interval}
Let $m$ be odd and let $I\subseteq\Z_m$ be a cyclic interval with at least
two residues outside it.  Suppose nonempty sets $P,Q\subseteq\Z_m$ satisfy
$P+Q=I$ and $|P+Q|=|P|+|Q|-1$, and suppose that $I$ is aperiodic.  Then $P$
and $Q$ are cyclic arithmetic progressions with a common difference.  Here a
\emph{cyclic arithmetic progression} with difference $d$ is a set
$\{a,a+d,\ldots,a+(k-1)d\}\subseteq\Z_m$.  Since $P+Q$ is the interval
$I$, their common difference is $1$ or $-1$.  Singletons are included as
progressions of length one.
\end{lemma}

\begin{lemma}\label{lem:round-robin-block}
Let $S\subseteq\{1,2,\ldots,q-2\}$ be a set of $b$ consecutive indices,
where $1\le b\le q-3$.  Then
$\bic(\bigcup_{t\in S}F_t)\le b+1$.
\end{lemma}

\begin{proof}
Let $G_S=\bigcup_{t\in S}F_t$.  Let $A$ and $B$ be disjoint vertex
subsets forming the two sides of a complete bipartite subgraph of $G_S$.

Assume first that $\infty\notin A\cup B$.  Put $X=A$ and $Y=B$.  The edge
  $xy$ belongs to $F_t$ exactly when $x+y=2t\pmod {q-1}$.  Since $q-1$ is odd,
  multiplication by two is a bijection of $\Z_{q-1}$.  We write
  $2^{-1}X=\{z\in\Z_{q-1}:2z\in X\}$ for the inverse image under this bijection.
Then $2^{-1}X+2^{-1}Y\subseteq S$, and \Cref{lem:interval} gives
$|A|+|B|\le b+1$.

Now suppose $\infty\in A$.  Write $A=\{\infty\}\cup X$ and $B=Y$.  The
neighbors of $\infty$ in the selected perfect matchings are exactly the
residues in $S$, so $Y\subseteq S$.  If $X=\varnothing$, then
$|A|+|B|=1+|Y|\le b+1$.  We may therefore assume $X\ne\varnothing$.
As before, $P+Q\subseteq S$, where
$P=2^{-1}X$ and $Q=2^{-1}Y$.  By \Cref{lem:interval},
$|X|+|Y|\le b+1$.  We show that equality is impossible.  Suppose
$|X|+|Y|=b+1$.
  Let $T=P+Q$ and let its stabilizer have size $h$.  Put $R=(q-1)-b\ge2$.  If $h>1$, then the complementary interval of $S$ either meets every stabilizer coset, which would force $T=\varnothing$, or it meets $R$ distinct cosets.  In the latter case Kneser's theorem gives
\[
 |P|+|Q|\le |T|+h
  \le (q-1)-h(R-1)
  \le (q-1)-R=b,
\]
  because $q-1$ is odd and hence $h\ge3$.  This contradicts the assumed equality.
Thus $T$ is aperiodic, and Kneser's theorem gives
$|T|\ge |P|+|Q|-1=b$.
Since $T\subseteq S$ and $|S|=b$, we have $T=S$.  Apply
\Cref{lem:critical-interval}.  First negate $P,Q$, and $S$ simultaneously
if their common difference is $-1$, meaning that every element $x$ is
replaced by $-x$.  Let $s_0$ be the first residue of the resulting interval
$S$, and translate both $P$ and $Q$ by $-s_0/2$, meaning that this value is
  added to every element; because $q-1$ is odd, $s_0/2$ is the unique residue
whose double is $s_0$.  The sumset is then translated by $-s_0$.  These
operations send $S$ to $\{0,1,\ldots,b-1\}$ and preserve both
$2Q\subseteq S$ and $P\cap Q=\varnothing$.  We may therefore write
$S=\{0,1,\ldots,b-1\}$ and
\[
 P=\{\alpha,\alpha+1,\ldots,\alpha+r-1\},
 \qquad
 Q=\{\beta,\beta+1,\ldots,\beta+\ell-1\},
\]
  where $r+\ell=b+1$ and $\alpha+\beta=0\pmod {q-1}$.
  Let $c\in\{0,\ldots,q-2\}$ satisfy $c\equiv2\beta\pmod {q-1}$.  Since
$Y=2Q=\{c,c+2,\ldots,c+2(\ell-1)\}\subseteq S$,
this step-two progression cannot wrap around the cyclic group: immediately
  before a wrap it would contain $q-3$ or $q-2$, while both residues lie
outside $S$.  Hence $c+2\ell-2\le b-1$, and therefore
  $r=b+1-\ell\ge c+\ell\ge c+1$.
The element $\alpha+c=-\beta+2\beta=\beta$ belongs to both $P$ and $Q$.
  Multiplication by two is a bijection of $\Z_{q-1}$, so this contradicts
  $X\cap Y=\varnothing$.  Thus $|X|+|Y|\le b$, and
  $|A|+|B|\le b+1$.
\end{proof}

The next elementary sumset lemma chooses how many paired perfect matchings and
cross-layer matchings each ordinary subsystem receives.
For a positive integer $r$, define
\[
 B_r=
 \begin{cases}
 \{0,1,\ldots,r\},&r\text{ is odd},\\[1mm]
 \{0,1,\ldots,r\}\setminus\{r/2\},&r\text{ is even}.
 \end{cases}
\]

\begin{lemma}\label{lem:even-counts}
Let $r_1,\ldots,r_m$ be positive integers, at least one of which is odd.  Then
\[
 B_{r_1}+\cdots+B_{r_m}
 =
 \{0,1,\ldots,r_1+\cdots+r_m\}.
\]
\end{lemma}

\begin{proof}
If $r_j$ is odd, then $B_{r_j}$ is a nontrivial full interval.  In every other \(B_r\), consecutive elements differ by at most two.
For every integer $M\ge1$, it is easy to see that
\(
 \{0,1,\ldots,M\}+B_r=\{0,1,\ldots,M+r\}.
\) Adding the sets one at a time gives the conclusion.
\end{proof}

\begin{proof}[Proof of \Cref{prop:even-factor}]
Put $r_i=d_i-1$.  Choose $s_i\in B_{r_i}$ with
$\sum_i s_i=q-2$; this is possible by \Cref{lem:even-counts}.
Since $\sum_i r_i=2q-d_*-1$, the number of cross-layer matchings to be
assigned is $\sum_i(r_i-s_i)=q-d_*+1$.  The choice
$s_i\in B_{r_i}$ ensures $s_i\ne r_i-s_i$, since the midpoint is omitted
when $r_i$ is even.

We may also assume $s_i\le q-3$ for every $i$.
Indeed, if some $s_i=q-2$, then $s_i\le r_i$ and $r_i=d_i-1\le q-2$
force $r_i=q-2$.  Since $\sum_i s_i=q-2$, all other $s_j$ are zero.
Choose an index $j\ne i$ for which $r_j$ is odd, and replace
$(s_i,s_j)$ by $(q-3,1)$.  This preserves the sum and remains admissible:
$q-3\in B_{q-2}$  and $1\in B_{r_j}$ because $r_j$
is odd.  All other $s_\ell$ remain zero, so every resulting $s_\ell$
is at most $q-3$.

Each of the $q-2$ paired perfect matchings is the union of
one perfect matching within $U$ and its corresponding perfect matching
within $V$.  Partition these full-vertex-set matchings into consecutive
blocks, assigning $s_i$ of them to subsystem $i$, and partition the
cross-layer matchings into consecutive blocks of sizes $r_i-s_i$.  Let
$G_i$ be the union of the two blocks.
A block of size zero is empty.  Otherwise,
\Cref{lem:round-robin-block} gives the layer biclique bound $s_i+1$, and
the shifts of the $r_i-s_i$ cross-layer matchings form a
cyclic interval $I$ of size $r_i-s_i$.  If $X\subseteq\Z_q$ and
$Y\subseteq\Z_q$ index the two sides of a biclique in this cross-layer
block, then every required edge has its shift in $I$, so $Y-X\subseteq I$.
Thus \Cref{lem:interval} gives $|X|+|Y|\le r_i-s_i+1$, which is the
cross-layer biclique bound.
Since $s_i\ne r_i-s_i$, \Cref{lem:two-layer} yields
$\bic(G_i)\le s_i+(r_i-s_i)+1=d_i$.
The blocks partition all remaining edges, completing the proof.
\end{proof}

Now we are ready to prove the main theorem for even $q$.

\begin{theorem}\label{thm:even-stabilization}
Let $p\ge3$ and $d_1,\ldots,d_p\ge2$. Suppose that the number of even local dimensions is positive and even, at least one local dimension is greater than two, and $q=(f_{LB}+1)/2$ is even.  Then
\[
 f_m(d_1,\ldots,d_p)=f_{LB}+1.
\]
\end{theorem}

\begin{proof}
The lower bound is \Cref{prop:lower}.  The known results
\Cref{thm:CJ,thm:22d-known} cover the case of $d_p\ge q$, where \(d_p\) is a dominant local dimension.
% For $(2,2,d)$, even $q=(d+3)/2$ means $d\equiv1\pmod4$, so the
% relevant result is Chen--Johnston~\cite[Theorem~3]{ChenJohnston2015}.
We may otherwise assume $d_p\le q-1$.

For $q\ge6$, choose an even exceptional dimension $d_*$.  
After removing it, an odd number of even dimensions remain, so \Cref{prop:even-factor} applies. 
 By \Cref{lem:bic-connect}, every ordinary graph $G_i$ has an orthogonal representation in $\mathbb R^{d_i}$, 
 in which every $d_i$ vectors are linearly independent.  
 By \Cref{prop:even-library}, $G_*$ is the orthogonality graph of $2q$
 exceptional local vectors in $\C^{d_*}$, every $d_*+1$ of which span
 $\C^{d_*}$.
  Since the exceptional and ordinary graphs  \(G_*,G_1,\ldots,G_m\) partition $K_{2q}$,  
  \Cref{lem:one-extra} shows that the resulting \(2q\) product states
form a UPB.

The only smaller value is $q=4$, for which the only remaining cases are
$(2^4,3)$ and $(2^2,3^2)$.  In either case, choose one qubit as the
exceptional subsystem.  Choose nonzero vectors $x,y,x',y'\in\C^2$ such that
$x\perp y$, $x'\perp y'$, and no two of them are
proportional.
For the eight exceptional local vectors, set
\[
 c_0=c_2=x,\qquad c_1=c_3=y,\qquad
 c_4=c_6=x',\qquad c_5=c_7=y'.
\]
Their orthogonality graph $G_*$ is exactly the disjoint union of the two
four-cycles $(0,1,2,3,0)$ and $(4,5,6,7,4)$.  Moreover, each of
$x,y,x',y'$ appears exactly twice among the $c_i$, so any three of the $c_i$
contain two nonproportional vectors and hence span $\C^2$.

The complement $\overline{G_*}$ is the disjoint union of the following five
perfect matchings, where $ij$ denotes the edge $\{i,j\}$:
\[
\begin{aligned}
 M_0&=\{02,13,46,57\},&
 M_1&=\{04,15,26,37\},\\
 M_2&=\{07,14,25,36\},&
 M_3&=\{05,16,27,34\},\\
 M_4&=\{06,17,24,35\}.
\end{aligned}
\]
Here $M_1\cup M_2$ and $M_3\cup M_4$ are the eight-cycles
\(
 (0,4,1,5,2,6,3,7,0)\) and \(
 (0,5,3,4,2,7,1,6,0),
\)
respectively.  Thus each $M_j$ is $1$-regular with biclique number $\bic(M_j)=2$,
while each displayed eight-cycle is $2$-regular with biclique number $3$.

For $(2^4,3)$, assign $M_1\cup M_2$ to the qutrit and assign
$M_0,M_3,M_4$ to the three ordinary qubits.  For $(2^2,3^2)$, assign
$M_1\cup M_2$ and $M_3\cup M_4$ to the two qutrits and assign $M_0$ to
the ordinary qubit.  The required ordinary local representations now follow
from \Cref{lem:bic-connect}. Combining with the $8$ exceptional local vectors $c_i$, \Cref{lem:one-extra} gives a UPB of size $2q=8$.
\end{proof}

%%%%%%%%%%%%%%%%%%%%%%%%%%%%%%%%%%%%%%%%%%%%
\section{Odd-\texorpdfstring{$q$}{q} Case}\label{sec:odd}

This section proves the upper bound when $q$ is odd.  We retain the two layers
$U=\{u_i:i\in\Z_q\}$ and $V=\{v_i:i\in\Z_q\}$ and the cross-layer
matchings $M_s$ defined in Section~\ref{sec:even}.  Put $h=(q-1)/2$.  For
$1\le a\le h$, define two $2$-regular layer graphs
\[
 F_a^U=\{u_i u_{i+a}:i\in\Z_q\},
 \qquad
 F_a^V=\{v_i v_{i+a}:i\in\Z_q\}.
\]
Thus $F_a^U$ joins each $u_i$ to $u_{i-a}$ and $u_{i+a}$, and $F_a^V$
does the same in the $V$ layer.  We call $a$ the corresponding
\emph{layer difference}.  

As before, one
subsystem is chosen as exceptional, $G_*$ denotes the orthogonality graph of
its local vectors, and all remaining subsystems are ordinary.  A layer graph
is \emph{available} if all its edges lie in $\overline{G_*}$, and a pair of
available layer graphs $F_a^U,F_a^V$ is \emph{aligned}.  
Let $A\subseteq\{1,\ldots,h\}$ consist of layer differences for which the
pairs $F_a^U,F_a^V$ are aligned, and let $B\subseteq\Z_q$ consist
of shifts for which the cross-layer matchings $M_b$ are unused by $G_*$.
Let $H(A,B)$ be the graph on $U\cup V$ formed by these aligned layer
pairs and unused cross-layer matchings.

The \emph{dihedral group of order $2q$} is
\[
 D_{2q}=\langle r,s\mid r^q=s^2=1,\ srs=r^{-1}\rangle.
\]
Its elements are the $q$ rotations $r^i$ and the $q$ reflections $r^is$,
$i\in\Z_q$.  
Identify $u_i$ with $r^i$ and $v_i$ with $r^is$.  Then $H(A,B)$ is the Cayley graph
of $D_{2q}$ with generating set
$\{r^{\pm a}:a\in A\}\cup\{r^bs:b\in B\}$.  That is, its vertex set is
$D_{2q}$, and each $x\in D_{2q}$ is joined to $xr^a$ and $xr^{-a}$ for
every $a\in A$, and to $xr^bs$ for every $b\in B$.  Under the
identification above, the first type of edges forms the aligned pair
$F_a^U,F_a^V$, while the second forms $M_b$.
The graph is undirected: for each $a\in A$, multiplication by $r^a$ can be
reversed by multiplication by $r^{-a}$, while $(r^bs)^2=1$ for each
$b\in B$, so multiplication by $r^bs$ reverses itself.  By Stong's theorem~\cite{Stong1985}, every connected Cayley graph on a
dihedral group admits a $1$-factorization; equivalently, its edges can be
partitioned into perfect matchings.

The proof has two stages.  First, \Cref{prop:odd-library} constructs the
exceptional local vectors and determines their orthogonality graph $G_*$ for
every possible non-qubit dimension.  Second, \Cref{prop:odd-allocation} partitions $\overline{G_*}$ among the
ordinary subsystems, using perfect-matching decompositions of the relevant
graphs $H(A,B)$. By
\Cref{lem:bic-connect}, the ordinary graphs have the required orthogonal
representations, and \Cref{lem:one-extra} then gives a UPB.

\begin{proposition}\label{prop:odd-library}
Let $q\ge5$ be odd and put $h=(q-1)/2$.
\begin{enumerate}[label=(\arabic*)]
\item If $q\ge7$, there are $2q$ vectors in $\R^3\subseteq \C^3$ such that
every four span $\C^3$.  Their orthogonality graph consists exactly of
$F_h^U$, $F_{h-1}^V$, and $M_0$.

\item There are $2q$ vectors in $\C^4$ such that every five span $\C^4$.
Their orthogonality graph consists exactly of
$F_h^U$, $F_{h-1}^V$, $M_0$, and $M_1$.

\item If $q\ge7$, there are $2q$ vectors in $\C^5$ such that every six
span $\C^5$.  In each layer, the orthogonality graph is
$F_h^U\cup F_{h-1}^U$ or $F_h^V\cup F_{h-1}^V$, and its cross-layer
edges form exactly $M_0$.  

\item For every $6\le d\le q-1$, there are $2q$ vectors in $\C^d$ such
that every $d+1$ span $\C^d$.  The graph induced by each layer is its
cycle $F_1^U$ or $F_1^V$, and the full orthogonality graph has cross-layer
degree $d-2$.  The unused cross-layer matchings form a cyclic interval of
length $q-d+2$.
\end{enumerate}
\end{proposition}

The proof is in
Appendix~\ref{app:odd-exceptional}.  Part~(1) uses a real double-cone
construction, and Part~(2) uses two Fourier orbits in a one-parameter
algebraic family.  Part~(3) relates two five-dimensional Fourier orbits by
a monomial unitary matrix, and Part~(4) uses weighted Fourier orbits whose
weights are obtained by the implicit function theorem.  Part~(1) is an
explicit construction, whereas (2)--(4) provide parameterized
families and establish the existence of suitable parameter choices.

\begin{proposition}\label{prop:odd-allocation}
Let $q\ge7$ be odd, and suppose $3\le d_p\le q-1$, 
\(
 \sum_{i=1}^p(d_i-1)+2=2q,
\)
and the number of even $d_i$ ($1\le i\le p$) is positive and even.  Choose
a subsystem of smallest non-qubit dimension as exceptional.  Write $t$ for its index and
$d_*=d_t$, and let $G_*$ be the orthogonality graph supplied by \Cref{prop:odd-library}.  Then $\overline{G_*}$ can be
partitioned into spanning subgraphs $G_i$  such that
\(
 \bic(G_i)\le d_i
\) for all $i\ne t$.
\end{proposition}

With the exceptional local vectors and the ordinary graphs now constructed,
the main theorem follows.

\begin{theorem}\label{thm:odd-stabilization}
Let $p\ge3$ and $d_1,\ldots,d_p\ge2$. Suppose that the number of even local dimensions is positive and even, at least one local dimension is greater than two, and $q=(f_{LB}+1)/2$ is odd.  Then
\[
 f_m(d_1,\ldots,d_p)=f_{LB}+1.
\]
\end{theorem}

\begin{proof}
The lower bound is \Cref{prop:lower}.  The known results
\Cref{thm:CJ,thm:22d-known} cover all dominant cases with $d_p\ge q$. 
We may therefore assume $d_i\le q-1$ for all
$1\le i\le p$.  Since some $d_i>2$ and $q$ is odd, this also forces
$q\ge5$.

First suppose $q\ge7$, and choose a subsystem of smallest non-qubit dimension as exceptional. For its dimension $d_*$,
\Cref{prop:odd-library} supplies $2q$ exceptional local vectors in
$\C^{d_*}$, every $d_*+1$ of which span, with orthogonality graph $G_*$.  By
\Cref{prop:odd-allocation}, the complement $\overline{G_*}$ can be
partitioned into ordinary graphs $G_i$ satisfying $\bic(G_i)\le d_i$.
Then \Cref{lem:bic-connect} gives ordinary local vectors in $\R^{d_i}$ for
which every $d_i$ are linearly independent.  Since $G_*$ and the graphs
$G_i$ partition $K_{2q}$, the resulting product states are pairwise
orthogonal, and \Cref{lem:one-extra} shows that they form a UPB of size
$2q=f_{LB}+1$.

It remains to treat $q=5$.  Then
$2+\sum_j(d_j-1)=2q=10$. The assumptions $d_i\le4$ leave the eight tuples
\[
\begin{gathered}
 (2^6,3),\quad
 (2^4,3^2),\quad
 (2^2,3^3),\quad
 (2^3,3,4),\\
 (2,3^2,4),\quad
 (3,4,4),\quad
 (2^5,4),\quad
 (2^2,4^2),
\end{gathered}
\]
where exponents denote multiplicities of local dimensions.
For the first six tuples, choose one qutrit ($3$-dim subsystem) as exceptional, and index the ten product
states by the $10$ edges
of $K_5$ with vertex set $\{0,1,2,3,4\}$.  Fix
$u\in\{0,1,2,3,4\}$, orient the edges arbitrarily, and, for an oriented
edge $e$, put $g_e=\mathbf e_{e^-}-\mathbf e_{e^+}\in\R^5$.  Set
$\rho=\mathbf e_u-\frac15\one$ and
$W=\one^\perp\cap\rho^\perp$.  Since $\dim W=3$, we identify $W$
isometrically with $\R^3\subseteq\C^3$.  Let $P:\one^\perp\to W$ be
the orthogonal projection and use $Pg_e$ as the exceptional local vector
indexed by $e$.  These vectors are nonzero: indeed,
$\rho\cdot\rho=4/5$ and $g_e\cdot\rho=g_e\cdot\mathbf e_u$, so
$\lVert Pg_e\rVert^2=2$ when $u\notin e$ and
$\lVert Pg_e\rVert^2=3/4$ when $u\in e$.  For distinct edges $e$ and
$f$, the projection formula gives
$(Pg_e)\cdot(Pg_f)=g_e\cdot g_f-\frac54
(g_e\cdot\mathbf e_u)(g_f\cdot\mathbf e_u)$.  This inner product is
zero when $e$ and $f$ are disjoint.  If they meet away from $u$, it is
$\pm1$, while if they meet at $u$, it is $\pm\frac14$; hence it vanishes
exactly when the two edges are disjoint.

It remains to check the required spanning property.  Let $S$ be any four
edges and let $h$ be the number of connected components of the spanning
subgraph $(V(K_5),S)$.  The oriented incidence vectors $\{g_e:e\in S\}$
span a subspace $U_S$ of dimension $5-h$, consisting of the vectors whose
coordinates sum to zero on each component.  Since a simple graph on five
vertices with at least three components has at most three edges, here
$h\le2$.  If $h=1$, then $U_S=\one^\perp$, so $P(U_S)=W$.  If $h=2$,
then $\dim U_S=3=\dim W$, and $P$ is injective on $U_S$: its kernel in
$\one^\perp$ is $\spann\{\rho\}$, whereas $\rho\notin U_S$ because the
sum of its coordinates on either nonempty proper component $C$ is
$-|C|/5$ when $u\notin C$ and $1-|C|/5$ when $u\in C$, both nonzero.
Thus $P(U_S)=W$ also in this case.  Consequently, any four exceptional
local vectors span $W\cong\R^3$ over $\R$, and hence $\C^3$ over $\C$.

Let $G_*$ be the orthogonality graph of these exceptional vectors.
Two vertices of $G_*$ are adjacent exactly when the
corresponding $K_5$ edges are disjoint, so $G_*$ is the Petersen graph.
The edge set of $\overline{G_*}$ is partitioned into the following six
perfect matchings, where $ab\! -\! cd$ denotes the edge joining the vertices
$ab$ and $cd$:
\[
\begin{aligned}
 Q_1={}&\{01\! -\!02,03\! -\!04,12\! -\!13,14\! -\!24,23\! -\!34\},\\
 Q_2={}&\{01\! -\!14,02\! -\!03,04\! -\!24,12\! -\!23,13\! -\!34\},\\
 Q_3={}&\{01\! -\!03,02\! -\!24,04\! -\!34,12\! -\!14,13\! -\!23\},\\
 Q_4={}&\{01\! -\!12,02\! -\!04,03\! -\!13,14\! -\!34,23\! -\!24\},\\
 Q_5={}&\{01\! -\!04,02\! -\!12,03\! -\!23,13\! -\!14,24\! -\!34\},\\
 Q_6={}&\{01\! -\!13,02\! -\!23,03\! -\!34,04\! -\!14,12\! -\!24\}.
\end{aligned}
\]
Assign blocks to the ordinary non-qubit subsystems as follows:
\[
\begin{array}{c|c}
\text{dimension tuple}&\text{assigned blocks}\\
\hline
(2^6,3)&\text{none}\\
(2^4,3^2)&Q_1\cup Q_3\\
(2^2,3^3)&Q_1\cup Q_3,\quad Q_2\cup Q_4\\
(2^3,3,4)&Q_1\cup Q_2\cup Q_3\\
(2,3^2,4)&Q_4\cup Q_6,\quad Q_1\cup Q_2\cup Q_3\\
(3,4,4)&Q_1\cup Q_2\cup Q_3,\quad Q_4\cup Q_5\cup Q_6.
\end{array}
\]
For each tuple, assign each matching not used in its displayed blocks
separately to one ordinary qubit. Each matching has biclique number two, each
two-matching block is a $10$-cycle and has biclique number
three, while every displayed three-matching block has biclique number four.
Thus \Cref{lem:bic-connect} gives the required ordinary local vectors.
Combining them with the exceptional local vectors, \Cref{lem:one-extra}
gives the UPBs.

For the last two tuples $ (2^5,4)$ and $(2^2,4^2)$, choose one
four-dimensional subsystem as exceptional and use Part~(2)
of \Cref{prop:odd-library} with $q=5$.  Its graph has layer edges $u_i u_{i\pm2}$ and $v_i v_{i\pm1}$, 
and cross-layer edges $u_i v_i$ and $u_i v_{i+1}$.  Its complement is the
disjoint union of the five perfect matchings
\[
\begin{aligned}
 P_1={}&\{u_0v_2,u_1u_2,u_3u_4,v_0v_3,v_1v_4\},\\
 P_2={}&\{u_0u_1,u_2v_0,u_3v_1,u_4v_3,v_2v_4\},\\
 P_3={}&\{u_0u_4,u_1v_0,u_2v_4,u_3v_2,v_1v_3\},\\
 P_4={}&\{u_0v_3,u_1v_4,u_2u_3,u_4v_1,v_0v_2\},\\
 P_5={}&\{u_0v_4,u_1v_3,u_2v_1,u_3v_0,u_4v_2\}.
\end{aligned}
\]
For $(2^5,4)$, assign these five matchings separately to the five qubits.
For $(2^2,4^2)$, assign $P_1\cup P_2\cup P_3$ to the ordinary
four-dimensional subsystem and
$P_4,P_5$ to the two qubits. Each matching has biclique number two, and the three-matching block has biclique number
four. Hence, \Cref{lem:one-extra,lem:bic-connect} completes the proof.
\end{proof}

\noindent
\emph{Completing the proof of \Cref{thm:stabilization}.}
The lower bound is \Cref{prop:lower}.  If $q=(f_{LB}+1)/2$ is even, apply \Cref{thm:even-stabilization}.  If $q$ is odd, apply \Cref{thm:odd-stabilization}.

The theorem shows that the excess above the natural lower bound remains stable
 with respect to increasing dimensions: once a multipartite system has any local dimension above two, the parity obstruction costs exactly one state and never more.

% \begin{proof}[Proof of \Cref{thm:classification}]
% The one-party statement is immediate.  The bipartite formula is \Cref{thm:CJ}.  Let $p\ge3$.  If $e=0$ or $e$ is odd, the natural bound is attained by \Cref{thm:AL}.  If $e>0$ is even and the system is not all-qubit, apply \Cref{thm:stabilization}.  The remaining systems are all-qubit and are given by Johnston's classification.
% \end{proof}

\section{Consequences and further questions}\label{sec:outlook}

The complete formula for the minimum UPB size reveals a simple global picture.  For every non-all-qubit multipartite system with $p\ge3$, the excess over the natural lower bound satisfies $f_m-f_{LB}\in\{0,1\}$: it equals one exactly when the number of even local dimensions is positive and even.  Thus, among multipartite systems, gaps larger than one are confined to the all-qubit family.

The result also determines the largest completely entangled subspace that can arise as the orthogonal complement of a UPB.  Indeed, the complement of a minimum UPB has dimension $\prod_{j=1}^{p}d_j-f_m(d_1,\ldots,d_p)$.  When $f_{LB}$ is attainable, this dimension reaches the general upper bound for completely entangled subspaces; in every remaining non-all-qubit multipartite case, it falls short by exactly one.

Several natural questions remain.  First, can one classify minimum UPBs, rather than only determine their cardinality, for example through their possible local orthogonality graphs?  Second, which dimension tuples admit minimum UPBs over $\R$, or with coordinates in a prescribed number field, and can such UPBs be constructed efficiently with explicit exact coordinates?  Third, for a fixed tensor-product space, which cardinalities between the minimum and the full dimension are attained by UPBs?

\paragraph{Acknowledgements} I would like to thank Yiwei Zhang for helpful discussions and for suggestions.

\paragraph{AI declaration}
 Within the author-defined graph-theoretical framework, OpenAI's GPT models, including
GPT 5.6 Sol and Codex, generated most of the detailed derivations and local-vector
constructions, and also assisted with language editing and consistency checking.  The author verified, selected, and
revised the AI-generated material and takes full responsibility for the
claims and proofs.

\paragraph{Conflict of interest and data availability declaration} The author declares no conflict of interest. No data were used in this work.

\newpage 

\appendix
\section*{Appendix}
Appendix~\ref{app:cyclic} gives the biclique bounds
used in the proofs of \Cref{prop:even-factor,prop:odd-allocation}.
Appendices~\ref{app:odd-exceptional} and~\ref{app:odd-allocation} contain
the proofs of \Cref{prop:odd-library,prop:odd-allocation}, respectively.
Finally, Appendix~\ref{app:theta} gives a constructive proof of the
known minimum-size formula for $(2,2,d)$ with odd $d\ge3$, independent of Theorem~\ref{thm:22d-known}.

\section{Biclique bounds}\label{app:cyclic}

We retain the biclique number defined in \Cref{sec:framework} and the
sumset, stabilizer, and periodicity terminology introduced in
\Cref{sec:even}.  Throughout this appendix, indices lie in
$\Z_q=\{0,1,\ldots,q-1\}$, with addition and subtraction modulo $q$.
We use Kneser's inequality \eqref{eq:kneser} and the cyclic-interval
bound in \Cref{lem:interval} to obtain the further biclique bounds below.

\subsection{Cyclic annuli}

The layer graphs use symmetric cyclic intervals rather than single
intervals.  We call
$D(a,b)=\{\pm a,\pm(a+1),\ldots,\pm b\}\subseteq\Z_q$, where
$1\le a\le b<q/2$, a \emph{cyclic annulus}; it is the union of two opposite
cyclic intervals.  Write $s=|D(a,b)|=2(b-a+1)$.

\begin{lemma}\label{lem:odd-annulus}
Let $q$ be odd.  If nonempty $X,Y\subseteq\Z_q$ satisfy
$Y-X\subseteq D(a,b)$, then $|X|+|Y|\le s+1$.
\end{lemma}

\begin{proof}
Let $E=Y-X$, let $H$ be its stabilizer, and write $|H|=h$ and $q=hg$.
By \eqref{eq:kneser-basic}, $|X|+|Y|\le |E|+h$.
If the stabilizer is trivial, the claim follows from $|E|\le s$.
Otherwise $h\ge3$, because $q$ is odd.

The complement of $D(a,b)$ is the union of two cyclic intervals of sizes
$A=2a-1$ and $B=q-2b-1$, whose total size is $R:=A+B=q-s$.
If either interval has size at least $g$, it meets every $H$-coset and $E$
must be empty, which is impossible because $X$ and $Y$ are
nonempty.  Otherwise the two intervals meet at least
$c\ge\max\{A,B\}\ge(R+1)/2$ distinct cosets.  Thus $|E|\le q-ch$, and
\[
 |X|+|Y|\le q-h(c-1)
 \le q-\frac{3(R-1)}2
 \le q-R+1=s+1.
\]
\end{proof}

For even $q$, the \emph{antipodal translation} sends $x$ to $x+q/2$.
The only possible loss occurs when the annulus is invariant under this
translation, meaning that adding $q/2$ does not change it.

\begin{lemma}\label{lem:even-annulus}
Let $q$ be even and let $X,Y\subseteq\Z_q$ be nonempty with
$Y-X\subseteq D(a,b)$.  If $a+b\ne q/2$, then $|X|+|Y|\le s+1$.  If
$a+b=q/2$, then $|X|+|Y|\le s+2$; moreover, equality implies
$X+q/2=X$ and $Y+q/2=Y$.
\end{lemma}

\begin{proof}
For the case $a+b\ne q/2$, use the proof of \Cref{lem:odd-annulus}.  Now
$R=A+B$ is even, and $A,B$ are odd.  The condition $a+b\ne q/2$ is
equivalent to $A\ne B$, so $\max\{A,B\}\ge R/2+1$.  For every nontrivial
stabilizer, $h\ge2$, and hence $h(c-1)\ge2(R/2)=R$.
Together with \eqref{eq:kneser-basic}, this gives the stronger bound $|X|+|Y|\le s$ whenever the stabilizer is nontrivial; the trivial case gives $s+1$.

Now suppose $a+b=q/2$.  Let $H_0=\{0,q/2\}$ and let
$\pi:\Z_q\to\Z_{q/2}$ be the \emph{quotient map}, which replaces every
antipodal pair $\{x,x+q/2\}$ by a single element.  The annulus $D(a,b)$ is
$H_0$-periodic, and
$\pi(D(a,b))$ is a cyclic interval of size $s/2$.  Put
$X'=X+H_0$ and $Y'=Y+H_0$.
Then $Y'-X'\subseteq D(a,b)$.  Applying \Cref{lem:interval} in the quotient
gives $|X'|/2+|Y'|/2\le s/2+1$,
which proves the claimed bound.  If equality holds for $X,Y$, enlarging to
$X'$ and $Y'$ cannot increase either set; hence both $X$ and $Y$ are
invariant under the antipodal translation.
\end{proof}

\subsection{Combining the two layers}

The next lemma combines one layer bound and one cross-layer bound.  The two layer graphs need not be identical; only the same numerical bound is required.

\begin{lemma}\label{lem:two-layer}
Let $G$ be a graph on disjoint $q$-vertex layers $U$ and $V$, and let $s,t$ be nonnegative integers.  Assume that every complete bipartite subgraph contained in one layer uses at most $s+1$ vertices, and every complete bipartite subgraph with one side in each layer uses at most $t+1$ vertices.  If $s\ne t$, then $\bic(G)\le s+t+1$.
\end{lemma}

\begin{proof}
Let $A,B$ be the sides of a complete bipartite subgraph and write
\[
 A_U=A\cap U,\quad A_V=A\cap V,
 \qquad
 B_U=B\cap U,\quad B_V=B\cap V.
\]
If exactly two of these sets are nonempty, the claim follows from one of the
assumed bounds.  If exactly three are nonempty, add the relevant layer bound
and cross-layer bound.  Their common nonempty set is counted twice, so
$|A|+|B|\le(s+1)+(t+1)-1=s+t+1$.

Suppose all four sets are nonempty.  Then
\[
\begin{aligned}
 |A_U|+|B_U|&\le s+1,
 &|A_V|+|B_V|&\le s+1,\\
 |A_U|+|B_V|&\le t+1,
 &|A_V|+|B_U|&\le t+1.
\end{aligned}
\]
Adding gives $|A|+|B|\le s+t+2$.  Equality would force all four local inequalities to be equalities.  The two layer equalities would give $|A|+|B|=2s+2$, while the cross equalities would give $|A|+|B|=2t+2$, contradicting $s\ne t$.
\end{proof}

When an annulus for even $q$ is invariant under the antipodal translation,
the layer bound may be $s+2$.  A positive cross-layer degree removes the
extra vertex.

\begin{lemma}\label{lem:central-two-layer}
Let $q$ be even.  In each layer, let the graph be the annulus $D(a,b)$ with
$a+b=q/2$ and degree $s$.  Across the layers, use a cyclic interval of
$t\ge1$ cross-layer matchings.  If $s+t\le q-2$ and $s\ne t$, then
$\bic(G)\le s+t+1$.
\end{lemma}

\begin{proof}
Use the notation from the proof of \Cref{lem:two-layer}.  If exactly two parts are nonempty, they form either a layer biclique or a cross-layer biclique.  The bounds $s+2\le s+t+1$ and $t+1\le s+t+1$ prove the claim.  If exactly three parts are nonempty, add the unique relevant layer bound and cross-layer bound.  If the layer pair uses at most $s+1$ vertices, their repeated nonempty part saves at least one.  If it uses $s+2$, \Cref{lem:even-annulus} says that both sides of that layer pair are antipodally invariant, so the repeated part has size at least two.  In both cases the total is at most $s+t+1$.

Assume all four parts are nonempty.  The two cross-layer bounds give
$|A|+|B|\le2(t+1)$, and the two layer bounds give
$|A|+|B|\le2(s+2)$.
If $t\le s-1$ or $t\ge s+3$, one of these inequalities is already at most
$s+t+1$.  The remaining possibilities are $t=s+1$ and $t=s+2$.  Since
$s+t\le q-2$ and $q$ is even, in either case $t\le q/2$.  A cyclic interval
of at most $q/2$ residues cannot contain two residues differing by $q/2$.
Thus no cross-layer vertex can be adjacent to an antipodal pair.  Since all
four parts are nonempty, neither layer pair can attain $s+2$; both are at
most $s+1$.  Hence $|A|+|B|\le2s+2\le s+t+1$.
\end{proof}

The only equality $s=t$ retained in the construction is the five-dimensional split $s=t=2$.

\begin{lemma}\label{lem:five-block}
In each layer, use one difference pair $\{\pm a\}$, and across the layers
use two consecutive cross-layer matchings.  If the layer cycles are not
four-cycles and $2a\not\equiv\pm1\pmod q$, then $\bic(G)\le5$.  The same
conclusion holds when $q$ is even and $2a=q/2$, so that the layer graph is a
union of four-cycles.
\end{lemma}

\begin{proof}
First assume that the layer cycles are not four-cycles.  Each layer graph and
the graph formed by the cross-layer matchings is two-regular and contains no
$K_{2,2}$ (the complete bipartite graph with two vertices on each side), so
every layer or cross-layer biclique uses at most three
vertices.  A six-vertex biclique would have all four layer parts nonempty and
would force all four bounds to be equalities.  A two-vertex neighbor set in
a layer has index difference $2a$, while a two-vertex neighbor set across two
consecutive matchings has index difference one.  This contradicts the
hypothesis $2a\not\equiv\pm1\pmod q$.

Now assume $q$ is even and $2a=q/2$.  A layer biclique may have four vertices.  With at most three nonempty layer parts, a four-vertex layer biclique has a repeated part of size two, so the total is at most $4+3-2=5$.  Suppose all four parts are nonempty and six vertices are used.  Both cross-layer bounds must then be tight, so each cross-layer pair has sizes one and two.  If a two-vertex cross-layer neighbor set is paired with a singleton in its own layer, it is also a two-vertex layer-neighbor set.  If instead both parts in that layer have size two, they form the two sides of a four-cycle and are again antipodal pairs.  In either case the same two indices must differ by one, because they are the neighbors supplied by two consecutive cross-layer matchings, and by $q/2$, because they are antipodal layer neighbors.  This is impossible.
\end{proof}

\section{Proof of Proposition~\ref{prop:odd-library}}\label{app:odd-exceptional}

We prove the four parts of \Cref{prop:odd-library} in increasing order of the
exceptional subsystem's dimension.  Throughout this appendix, the vertices
$u_i,v_i$ retain the meaning from \Cref{sec:odd}, and their corresponding
exceptional local vectors are denoted by $c_{u_i},c_{v_i}$.

\subsection{Dimension three: the double-cone construction}

Let $q\ge7$ be odd and put $x=\pi/q$, $R=\cos x$, $S=\cos3x$, and
$\delta=\arccos\sqrt{RS}$.
Here $R>S>0$.  We first record
\begin{equation}\label{eq:delta-window}
 \sqrt5\,x<\delta<3x.
\end{equation}
The upper bound follows from $\sqrt{RS}>S=\cos3x$.  For the lower bound, the
function $g(t)=\log\cos\sqrt t$ is strictly concave on the relevant
interval, meaning that its graph lies strictly above every chord, because,
with $u=\sqrt t$, one has
$g''(t)=-(u\sec^2u-\tan u)/(4u^3)<0$.
Applying strict concavity at the midpoint of $x^2$ and $9x^2$ gives
\[
 \log\cos(\sqrt5x)
 >\frac{\log\cos x+\log\cos3x}{2}
 =\log\cos\delta.
\]
Since $\log\cos t$ is decreasing, the lower bound in \eqref{eq:delta-window} follows.  In particular,
\begin{equation}\label{eq:delta-noninteger}
 \delta/x\notin\Z.
\end{equation}

For $i\in\Z_q$, define
\[
\begin{aligned}
 c_{u_i}&=\bigl(\sqrt R,\cos(2ix),\sin(2ix)\bigr),\\
 c_{v_i}&=\bigl(-\sqrt S,\cos(2ix+\delta),\sin(2ix+\delta)\bigr).
\end{aligned}
\]
The two families lie on two circular cones, which explains the name of the
construction.
Let $h=(q-1)/2$.  Direct calculation gives
\[
 c_{u_i}\perp c_{u_j}
 \quad\Longleftrightarrow\quad
 j-i\equiv\pm h\pmod q,
\]
and
\[
 c_{v_i}\perp c_{v_j}
 \quad\Longleftrightarrow\quad
 j-i\equiv\pm(h-1)\pmod q.
\]
For a cross-layer pair,
$\inner{c_{u_i}}{c_{v_j}}=-\sqrt{RS}+\cos(2(j-i)x+\delta)$.
One solution is $i=j$.  Any other solution would imply that $\delta/x$ is an
integer modulo $q$, contrary to \eqref{eq:delta-noninteger}.  Thus the
cross-layer graph is exactly the matching $u_iv_i$.

It remains to prove that every four vectors span $\R^3$.  Any three vectors
in one layer are independent: after normalizing their common first
coordinate, dependence would make three distinct points on a circle
collinear, which is impossible because a line meets a circle in at most two
points.

Consider two vectors from each layer.  Write $\theta_i=2ix$.  A
\emph{normal} to a plane is a nonzero vector perpendicular to every vector
in that plane.  A normal to the plane generated by $c_{u_i},c_{u_j}$ is
\[
 n_U=\bigl(\cos\alpha,-\sqrt R\cos\mu,-\sqrt R\sin\mu\bigr),
 \quad
 \alpha=\frac{\theta_i-\theta_j}{2},
 \quad
 \mu=\frac{\theta_i+\theta_j}{2}.
\]
A normal to the plane generated by $c_{v_k},c_{v_\ell}$ is
\[
 n_V=\bigl(\cos\beta,\sqrt S\cos\nu,\sqrt S\sin\nu\bigr),
 \quad
 \beta=\frac{\theta_k-\theta_\ell}{2},
 \quad
 \nu=\frac{\theta_k+\theta_\ell}{2}+\delta.
\]
If the four vectors failed to span, these normals would be parallel.  Their
last two coordinates would then force $\nu-\mu\in\pi\Z$, or
$\delta+x(k+\ell-i-j)\in\pi\Z$.
This again makes $\delta/x$ an integer, a contradiction.  Thus every four
vectors span $\R^3$ over $\R$, and hence span $\C^3$ over $\C$, proving
Part~(1) of \Cref{prop:odd-library}.

\subsection{Dimension four}

The construction uses two weighted Fourier orbits, as defined in the proof of
\Cref{prop:even-library}, whose layer kernels have different
roots.  A construction depending on one parameter places the two required
cross-layer orthogonality matchings while keeping every five local vectors
spanning.

Let $q\ge5$ be odd and put $\theta=\pi/q$, $\omega=e^{2\pi i/q}$,
$\xi=e^{i\theta}$, $\rho=2\cos\theta$, and $\sigma=2\cos3\theta$.
Choose a positive real number $a$ sufficiently close to one, but not equal to one.  Define
\begin{equation}\label{eq:four-abR}
 b=\frac{a+\rho}{1+a\rho},
 \qquad
 R^2=\frac{1+a\rho}{a(a+\rho)}.
\end{equation}
Then $R^2ab=1$.
Set
\begin{equation}\label{eq:four-M}
 M(a)=
 \frac{R^2(\rho+a)(\sigma+b)-1-\rho^2R^2}{\rho}.
\end{equation}
Choose a nonreal solution $z$ of
\begin{equation}\label{eq:four-z}
 |z|^2=R^2,
 \qquad
 z+\overline z=M(a).
\end{equation}

Such a solution exists whenever $a$ is sufficiently close to, but different
from, one.  At $a=1$ we have $R=1$ and
$M(1)/2=(\rho^3+\rho^2-4\rho-2)/2$.
For $2\cos(\pi/5)\le\rho<2$, this number lies strictly between $-1$ and $1$.  The inequality remains strict under a small perturbation of $a$, so \eqref{eq:four-z} has a nonreal conjugate pair.  Since $a\ne1$, one also has $R\ne1$.

Put $\eta=\xi z$ and define
\[
 p(T)=(T-1)(T-\omega)(T-\eta)
     =\sum_{k=0}^{3}p_kT^k.
\]
A direct expansion gives
\[
 p_0=-\xi^3z,\qquad
 p_1=\xi^2(1+\rho z),\qquad
 p_2=-\xi(\rho+z),\qquad
 p_3=1.
\]

Define positive weights $\lambda_0=1$, $\lambda_1=\rho+a$,
$\lambda_2=1+\rho a$, and $\lambda_3=a$, together with $\mu_0=1$,
$\mu_1=\sigma+b$, $\mu_2=1+\sigma b$, and $\mu_3=b$.
For $a$ close to one all these numbers are positive.  Equations
\eqref{eq:four-abR}--\eqref{eq:four-z} imply
\begin{equation}\label{eq:four-moduli}
 |p_k|^2=R^2\lambda_k\mu_k
 \qquad(0\le k\le3).
\end{equation}
For $k=0,3$ this is immediate.  The equality for $k=1$ is
\eqref{eq:four-M}; the equality for $k=2$ follows from
$|\rho+z|^2-|1+\rho z|^2=(\rho^2-1)(1-R^2)$ and
\eqref{eq:four-abR}.

For $i\in\Z_q$, define
\begin{equation}\label{eq:four-vectors}
 c_{u_i}=
 \bigl(\sqrt{\lambda_k}\,\omega^{ki}\bigr)_{k=0}^{3},
 \qquad
 c_{v_i}=
 \left(\frac{p_k}{\sqrt{\lambda_k}}\,
       \omega^{ki}\right)_{k=0}^{3}.
\end{equation}

The first layer kernel factors as
\[
 \sum_{k=0}^{3}\lambda_kT^k
 =
 (1+aT)(1+\rho T+T^2).
\]
The roots of the quadratic factor are
\(
 \omega^{(q-1)/2}\) and \(
 \omega^{-(q-1)/2},
\)
while $-1/a$ is not a $q$th root of unity.  Hence
\[
 c_{u_i}\perp c_{u_j}
 \quad\Longleftrightarrow\quad
 j-i\equiv\pm\frac{q-1}{2}\pmod q.
\]
By \eqref{eq:four-moduli},
\[
 \sum_{k=0}^{3}\mu_kT^k
 =
 (1+bT)(1+\sigma T+T^2),
\]
so
\[
 c_{v_i}\perp c_{v_j}
 \quad\Longleftrightarrow\quad
 j-i\equiv\pm\frac{q-3}{2}\pmod q.
\]
Finally, $\inner{c_{u_i}}{c_{v_j}}=p(\omega^{j-i})$.
The roots $1$ and $\omega$ give the two cross-layer matchings with shifts
$0$ and $1$.  The third root $\eta$ is not a root of unity because
$|\eta|=R\ne1$.  Thus the prescribed orthogonality graph is four-regular.

It remains to choose the parameter so that every five local vectors span.
Any four vectors within one layer are independent by the Vandermonde
determinant.  Only selections of type $3U+2V$ and $2U+3V$, indicating the
number selected from each layer, need attention.

For the polynomial argument, temporarily regard $z$ and $\overline z$ as
independent variables constrained by polynomial equations.  Eliminating
$\overline z$ from \eqref{eq:four-z} gives
$z^2-M(a)z+R^2(a)=0$.
The \emph{discriminant} of a quadratic $Az^2+Bz+C$ is $B^2-4AC$; it
determines whether the quadratic factors over a given field.  Here the
discriminant can be written as
\begin{equation}\label{eq:four-discriminant}
 \Delta(a)=
 \frac{\rho(a^2+\rho a+1)Q_\rho(a)}
      {a^2(a+\rho)^2},
\end{equation}
where
\[
 Q_\rho(a)=
 \rho(\rho^2-3)^2(a^2+1)
 +a(\rho^6-6\rho^4+5\rho^2+8).
\]
Moreover,
\[
 Q_\rho(a)\equiv-4a(\rho^2-2)
 \pmod{a^2+\rho a+1}.
\]
The congruence means that the difference of the two polynomials is divisible
by $a^2+\rho a+1$.  This polynomial is \emph{squarefree}, meaning that it has
no repeated root, because $\rho^2\ne4$.  It is also coprime to $Q_\rho$,
meaning that they have no common root: at either nonzero root of
$a^2+\rho a+1$, the displayed remainder is nonzero because $q$ is odd and
hence $\rho^2\ne2$.  Consequently these roots occur to an odd power in the
factorization of the discriminant \eqref{eq:four-discriminant}, so the
discriminant is not a square in the field $\C(a)$ of rational functions in
$a$, namely quotients of polynomials in $a$.  The quadratic therefore cannot
be factored over $\C(a)$; equivalently,
its solution set is an \emph{irreducible algebraic curve}, meaning that it
cannot be written as the union of two smaller solution sets defined by
polynomial equations.

Call the real values of $a$ near one together with the nonreal solutions
$z$ of \eqref{eq:four-z} the \emph{admissible parameter choices}.  Suppose
that a fixed five-element selection failed to span $\C^4$ for every
admissible choice in an open interval.  Dividing row $k$ of
all the vectors by the nonzero factor $\sqrt{\lambda_k}$ does not affect
whether any relevant $4\times4$ minor vanishes.  After this row normalization,
the entries are $\omega^{ki}$ in a $U$-column and
$p_k\omega^{ki}/\lambda_k$ in a $V$-column.  Thus every relevant minor is a
rational function of $a$ and $z$ on the irreducible curve.  The admissible
choices trace a
nonconstant real-analytic arc, meaning that they vary locally as convergent
power series in one real parameter.  A nonzero rational function on an
irreducible curve cannot vanish along such an arc, so all these minors would
have to vanish identically on the curve.

We test this conclusion on the \emph{bounded branch} as $a\to0$, namely the
solution $z(a)$ that remains bounded.  From \eqref{eq:four-abR} and
\eqref{eq:four-M},
\[
 R^2=\frac1{\rho a}+O(1),
 \qquad
 M(a)=\frac{\sigma}{\rho a}+O(1).
\]
The two roots of \eqref{eq:four-z} have product $R^2$ and sum $M(a)$;
consequently one root diverges and the other satisfies
$z=1/\sigma+O(a)$.
Work on this bounded branch.  Divide row $k$ of \eqref{eq:four-vectors} by
$\sqrt{\lambda_k}$ and multiply every $V$-column by $a$.  For a
root-of-unity node $w$, the normalized $U$-layer column is
$U(w)=(1,w,w^2,w^3)^T$,
and a direct Laurent expansion, that is, a power-series expansion allowing
finitely many negative powers of $a$, gives
\begin{equation}\label{eq:four-limit-v}
 aV(w)
 =
 w^3e_3+
 a(\alpha_0,\alpha_1w,\alpha_2w^2,0)^T
 +O(a^2),
\end{equation}
where
\[
 \alpha_0=-\frac{\xi^3}{\sigma},
 \qquad
 \alpha_1=\frac{\xi^2(\rho+\sigma)}{\rho\sigma},
 \qquad
 \alpha_2=-\frac{\xi(\rho\sigma+1)}{\sigma}.
\]

For a $3U+2V$ selection, three distinct $U$-columns and either limiting
$V$-column already have a nonzero determinant.  Expanding along the last
coordinate of the limiting $V$-column leaves the three-point Vandermonde
determinant.  Hence no such selection fails to span for every parameter
choice.

Consider a $2U+3V$ selection.  Let the two $U$-nodes be $r_1,r_2$ and the
three $V$-nodes be $s_1,s_2,s_3$.  If these five columns always spanned a
space of dimension at most three, choose a nonzero \emph{left-null vector},
meaning a row vector
whose product with every selected column is zero, with entries given by
Laurent series in $a$.  Normalize its first nonzero term as
$\ell(a)=\ell_0+a\ell_1+O(a^2)$.
The leading $V$-columns in \eqref{eq:four-limit-v} force the fourth coordinate of $\ell_0$ to vanish.  The two $U$-constraints then say that the polynomial represented by the first three coordinates of $\ell_0$ is
\[
 L_0(T)=\kappa(T-r_1)(T-r_2),
 \qquad \kappa\ne0.
\]
Let $\gamma$ be the fourth coordinate of $\ell_1$.  The coefficient of $a$
in the three $V$-constraints says that each $s_j$ is a root of
\[
 F(T)=
 \gamma T^3
 +\kappa\alpha_2T^2
 -\kappa\alpha_1(r_1+r_2)T
 +\kappa\alpha_0r_1r_2.
\]
The constant term is nonzero, so $F$ is not the zero polynomial.  Since it has the three distinct roots $s_1,s_2,s_3$, it equals
\(
 \gamma(T-s_1)(T-s_2)(T-s_3),
\)
and in particular $\gamma\ne0$.  Comparing the quadratic and constant coefficients, without dividing by any parameter-dependent quantity, gives
\(
 \alpha_2s_1s_2s_3
 =\alpha_0r_1r_2(s_1+s_2+s_3).
\)
Substituting the values of $\alpha_0$ and $\alpha_2$ yields the cross-multiplied identity
\begin{equation}\label{eq:four-cross-multiplied}
 (\rho\sigma+1)s_1s_2s_3
 =\xi^2r_1r_2(s_1+s_2+s_3).
\end{equation}

Assume first that $q\ge7$.  Then $\rho\sigma+1=\sin(5\theta)/\sin\theta>0$.  Taking absolute values in \eqref{eq:four-cross-multiplied} gives
\begin{equation}\label{eq:four-root-sum}
 |s_1+s_2+s_3|
 =\rho\sigma+1
 =\frac{\sin5\theta}{\sin\theta}.
\end{equation}
For $q\ge9$, the sum of three distinct $q$th roots has modulus at most
\(
 1+2\cos\frac{2\pi}{q}
 =\frac{\sin3\theta}{\sin\theta}.
\)
Indeed, rotate the sum to the positive real axis.  Its modulus is the sum of the real projections of the three selected roots, so the three largest projections are obtained from three consecutive grid points.  If the middle of these points makes angle $\delta$ with the positive axis, their projection sum is
\[
 (1+2\cos(2\pi/q))\cos\delta
 \le 1+2\cos(2\pi/q).
\]
Since
\(
 \sin5\theta-\sin3\theta
 =2\cos4\theta\sin\theta>0,
\)
this contradicts \eqref{eq:four-root-sum}.

For $q=7$, put $c_j=\cos(2j\pi/7)$.  Up to rotation and reflection, the three exponents have one of the four forms
\[
 \{0,1,2\},\quad \{0,1,3\},\quad
 \{0,1,4\},\quad \{0,2,4\}.
\]
The corresponding squared moduli are
\[
 3+4c_1+2c_2,
 \quad 2,
 \quad 3+2c_1+4c_3,
 \quad 3+4c_2+2c_3.
\]
The target in \eqref{eq:four-root-sum} is
\(
 (\rho\sigma+1)^2=2+2c_1.
\)
Using $1+2(c_1+c_2+c_3)=0$, together with $c_1>0>c_2>c_3$ and $c_3<-1/2$, the four differences from the target are respectively
\(
 -2c_3, -2c_1, 1+4c_3, 2(c_2-2c_1),
\)
all nonzero.  Thus \eqref{eq:four-root-sum} is impossible for $q=7$ as well.

Finally, when $q=5$ one has $\rho\sigma+1=0$.  Equation~\eqref{eq:four-cross-multiplied} then gives
\(
 s_1+s_2+s_3=0,
\)
because $\xi^2r_1r_2\ne0$.  Three distinct fifth roots cannot have zero sum: otherwise the remaining two fifth roots would also sum to zero, which would require their ratio to be $-1$.  This is impossible in a group of odd order.

Thus no fixed five-element selection fails to span for every parameter
choice.  Each selection excludes only finitely many points of the admissible
parameter arc, and there are finitely many selections.  One admissible
parameter choice avoids all
of them.  For this choice every five vectors span, proving Part~(2) of
\Cref{prop:odd-library}.

\subsection{Dimension five}

We prove Part~(3) of \Cref{prop:odd-library}.  Let $q\ge7$ be odd, put
$\omega=e^{2\pi i/q}$, $h=(q-1)/2$, and $\theta=\pi/q$, and define
\begin{equation}\label{eq:five-a-b}
 a=2(\cos\theta+\cos3\theta),
 \qquad
 b=2+4\cos\theta\cos3\theta.
\end{equation}
Both numbers are positive.  The polynomial
\[
 K(z)=(z^2+2\cos\theta\,z+1)
      (z^2+2\cos3\theta\,z+1)
 =z^4+az^3+bz^2+az+1
\]
has roots $\omega^{\pm h}$ and $\omega^{\pm(h-1)}$.

Let $\lambda_0=\lambda_4=1$, $\lambda_1=\lambda_3=a$, and
$\lambda_2=b$, and define
$c_{u_i}=(\sqrt{\lambda_k}\,\omega^{ki})_{k=0}^{4}$.  Then
$\inner{c_{u_i}}{c_{u_j}}=K(\omega^{j-i})$,
so the layer orthogonality differences are precisely $\pm h,\pm(h-1)$.  Any
five $U$-layer local vectors are independent by the Vandermonde determinant.

We next construct a second orbit.  From \eqref{eq:five-a-b},
$4b-a^2=4(2-(\cos\theta-\cos3\theta)^2)>0$.
Choose $\phi\in(0,\pi/2)$ such that
\begin{equation}\label{eq:phi-choice}
 \cos\phi=\frac{a}{2\sqrt b}.
\end{equation}
For $t$ on the unit circle, meaning $|t|=1$, define the monomial unitary
$R_t$ by
\(
 R_te_k=r_ke_{k+2\pmod5},
\)
where
\[
 r_0=e^{i\phi},
 \quad r_1=-1,
 \quad r_2=e^{-i\phi},
 \quad r_3=t,
 \quad r_4=-t.
\]
A \emph{monomial unitary} maps every standard basis vector to a
unit-modulus multiple of another standard basis vector; equivalently, its
matrix has exactly one nonzero entry of modulus one in each row and column.
Put $c_{v_i}=R_tc_{u_i}$.  Each layer has the same orthogonality graph, and
any five $V$-layer local vectors are independent.

Equation~\eqref{eq:phi-choice} gives
\begin{equation}\label{eq:five-matching}
 \inner{c_{u_i}}{c_{v_i}}
 =\omega^{-2i}\bigl(\sqrt b e^{i\phi}-a+\sqrt b e^{-i\phi}\bigr)=0.
\end{equation}
For $\Delta=j-i\ne0$,
\begin{equation}\label{eq:five-cross-generic}
\begin{aligned}
 \inner{c_{u_i}}{c_{v_j}}
 &=\omega^{-2i}
   \bigl(\sqrt b e^{i\phi}-a\omega^\Delta
        +\sqrt b e^{-i\phi}\omega^{2\Delta}\bigr)\\
 &\quad+t\sqrt a\,\omega^{3i+3\Delta}(1-\omega^\Delta).
\end{aligned}
\end{equation}
The coefficient of $t$ is nonzero, so every cross-layer pair outside the
matching excludes at most one value of $t$.  Outside a finite set of
parameter values, the only cross-layer orthogonality is the matching in
\eqref{eq:five-matching}.

We now impose the spanning property.  For a unit complex number $z$, write
\[
 c_U(z)=
 \begin{pmatrix}
 1\\ \sqrt a\,z\\ \sqrt b\,z^2\\ \sqrt a\,z^3\\ z^4
 \end{pmatrix},
 \qquad
 c_{V,t}(z)=
 \begin{pmatrix}
 t\sqrt a\,z^3\\ -tz^4\\ e^{i\phi}\\ -\sqrt a\,z\\ e^{-i\phi}\sqrt b\,z^2
 \end{pmatrix}.
\]
Fix six distinct columns chosen from the two orbits, with $s$ from the first.  If $s\in\{0,1,5,6\}$, five columns already come from one Vandermonde orbit.  Suppose $s\in\{2,3,4\}$.

Failure to span is an algebraic condition in $t$, meaning that it is given
by polynomial equations in $t$.  If it held for infinitely many $t$, all
relevant $5\times5$ minors would vanish identically, that is, for every
value of $t$.  For
$s=2$ or $3$, set $t=0$.  At least three selected $V$-layer columns then
have support, meaning nonzero coordinates, only in the last three positions.
A left-null vector would define a polynomial of degree at most two vanishing
at three distinct roots of unity, so its last three coordinates would be
zero.  The remaining at least two $U$-layer columns would then force its first
two coordinates to be zero.

For $s=4$, divide the two selected $V$-layer columns by $t$ and let
$t\to\infty$.  Their limit has support only in the first two coordinates and
gives a degree-one polynomial at two distinct nonzero nodes, forcing the first
two coordinates of a left-null vector to vanish.  The four $U$-layer columns
then give a degree-two polynomial at four distinct nodes, forcing the last
three coordinates to vanish.  Both cases are contradictions.

Thus every six-element selection has only finitely many bad values of $t$.
Choose one unit-modulus $t$ outside their union and outside the finite set
excluded by \eqref{eq:five-cross-generic}.  Then every six local vectors span.  The
layer graph uses the four differences $\pm h,\pm(h-1)$, leaving
$1,\ldots,h-2=(q-5)/2$, and the only used cross-layer matching is
$\{u_iv_i:i\in\Z_q\}$.
This proves Part~(3) of \Cref{prop:odd-library}.

\subsection{Dimensions at least six}

We use the weighted Fourier-orbit notation and determinant formula from the
proof of \Cref{prop:even-library}.  Let $q$ be odd and
$6\le d\le q-1$, and put $h=(q-1)/2$, $\zeta=e^{2\pi i h/q}$, and
$\rho=2\cos(\pi/q)$.  Since $h$ and $q$ are coprime, $\zeta$ is a primitive
$q$th root of unity, and
\begin{equation}\label{eq:three-term-general}
 1+\rho\zeta+\zeta^2=0.
\end{equation}
Let $r=q-d$ and define
$P(z)=\prod_{h'=d}^{q-1}(z-\zeta^{h'})$.
Equation~\eqref{eq:q-binomial}, with the present primitive root, again shows that every coefficient of $P$ is nonzero.  Introduce a complex parameter $\eta$ and set
\begin{equation}\label{eq:odd-p-eta}
 p_\eta(z)=P(z)(z-\eta).
\end{equation}
Outside a finite set of values of $\eta$, all $r+2$ coefficients of
$p_\eta$ are nonzero and $p_\eta(\zeta^j)\ne0$ for $0\le j<d$.
Fix such an $\eta$ temporarily and write $p_j=p_\eta(\zeta^j)$ and $A_j=|p_j|^2$.

Put $M=d-5$.  Use $\alpha_0=1$, $\alpha_1=\rho$, $\alpha_2=1$ and
$\beta_{d-3}=1$, $\beta_{d-2}=\rho$, $\beta_{d-1}=1$.
Both triples satisfy shifted forms of \eqref{eq:three-term-general}.  Set
$g_j=A_j/\beta_j$ for $d-3\le j\le d-1$, and define
\[
\begin{aligned}
 \lambda_0&=1+x_0,
 &\lambda_1&=\rho+x_1,
 &\lambda_2&=1,\\
 \lambda_j&=\varepsilon^{j-2} &&(3\le j\le d-4),\\
 \lambda_{d-3}&=\varepsilon^M(g_{d-3}+y_0),
 &\lambda_{d-2}&=\varepsilon^M(g_{d-2}+y_1),
 &\lambda_{d-1}&=\varepsilon^M g_{d-1}.
\end{aligned}
\]
We choose four real variables so that
\begin{equation}\label{eq:odd-balance-general}
 \sum_{j=0}^{d-1}\lambda_j\zeta^j=0,
 \qquad
 \sum_{j=0}^{d-1}\frac{A_j}{\lambda_j}\zeta^j=0.
\end{equation}
After multiplying the second complex equation by $\varepsilon^M$ and cancelling explicit powers, the four real equations extend smoothly to $\varepsilon=0$.  Their Jacobian in $(x_0,x_1,y_0,y_1)$ is block diagonal.  The first block has columns $1,\zeta$; the second has nonzero multiples of $\zeta^{d-3},\zeta^{d-2}$.  Each pair is linearly independent over $\R$, so the implicit function theorem yields positive weights for small $\varepsilon$.

Define $c_{u_i},c_{v_i}$ by \eqref{eq:even-fourier-vectors}, with the
current $\zeta,\lambda_j$ and with $P(\zeta^j)$ replaced by $p_j$.
Equation~\eqref{eq:odd-balance-general}
gives $c_{u_i}\perp c_{u_{i+1}}$ and
$c_{v_i}\perp c_{v_{i+1}}$.
These prescribed cycle edges are the only layer orthogonalities for sufficiently small $\varepsilon$.  For a $q$th root $z$, the two layer kernels satisfy
\[
 \sum_{j=0}^{d-1}\lambda_jz^j
 \longrightarrow1+\rho z+z^2,
\]
and
\[
 \varepsilon^M\sum_{j=0}^{d-1}\frac{A_j}{\lambda_j}z^j
 \longrightarrow z^{d-3}(1+\rho z+z^2).
\]
Among the $q$th roots of unity, the quadratic factor vanishes exactly at $z=\zeta$ and $z=\zeta^{-1}$.  Equation~\eqref{eq:odd-balance-general} makes the values at $z=\zeta$ exactly zero; because the coefficients are real, the values at $z=\zeta^{-1}$ are zero as well.  The remaining finitely many values stay nonzero after $\varepsilon$ is chosen sufficiently small.

The finite Fourier-inversion identity explained in the proof of
\Cref{prop:even-library}, applied to \eqref{eq:odd-p-eta}, shows
that a cross-layer inner product can be nonzero only when
$\ell-i\in\{0,-1,\ldots,-(r+1)\}$.
All coefficients of $p_\eta$ are nonzero, so these are precisely the
nonorthogonal cross-layer shifts.  The prescribed orthogonality graph
consequently has layer degree two and cross-layer degree
$q-(r+2)=d-2$.

The mixed determinant has the form of \eqref{eq:mixed-det-general}, with
$P(\zeta^k)$ replaced by $p_k$.  In the sense
defined in the proof of \Cref{prop:even-library}, the orders of the weights are
$0,0,0,1,2,\ldots,d-6,d-5,d-5,d-5$.
For $3\le s\le d-3$, the unique minimum-order row set is $\{0,\ldots,s-1\}$, and both associated Fourier minors are nonzero.

Two remaining layer counts require separate arguments.  For $s=1$, the
leading coefficient has the
form
\[
 \sum_{k=0}^{2}\gamma_k
 \frac{\alpha_k}{P(\zeta^k)(\zeta^k-\eta)}.
\]
Here $\gamma_0\ne0$, because deleting row $0$ leaves the consecutive rows
$1,\ldots,d-1$.  As $\eta\to1$, the $k=0$ term becomes unbounded while the
other two remain finite.  Hence the displayed expression is not identically
zero.

For $s=d-1$, write $p_k=P(\zeta^k)(\zeta^k-\eta)$.  The minimum-order
terms in \eqref{eq:mixed-det-general} are obtained by omitting one row
$k\in\{d-3,d-2,d-1\}$ and contain the factor
$\prod_{j\ne k}\lambda_j/p_j$.  After extracting the common nonzero factor
$\prod_j\lambda_j/p_j$, the remaining factor is $p_k/\lambda_k$.  Since
\[
 \lambda_k=\varepsilon^M
 \left(\frac{|p_k|^2}{\beta_k}+o(1)\right),
 \qquad d-3\le k\le d-1,
\]
we have
\[
 \frac{p_k}{\lambda_k}
 =\varepsilon^{-M}\frac{\beta_k}{\overline{p_k}}
  +o(\varepsilon^{-M}).
\]
Thus, after a common nonzero factor is removed, the leading coefficient has
the form
\[
 \sum_{k=d-3}^{d-1}\delta_k
 \frac{\beta_k}{\overline{P(\zeta^k)(\zeta^k-\eta)}}.
\]
The coefficient $\delta_{d-1}$ is nonzero because deleting the last row
leaves the consecutive rows $0,\ldots,d-2$.  As
$\overline\eta\to\overline{\zeta^{d-1}}$, the term $k=d-1$ becomes
unbounded while the other terms remain finite.  Hence this expression is not
identically zero.

For each selected column set, the $s=1$ coefficient is rational in $\eta$,
whereas the $s=d-1$ coefficient is rational in $\overline\eta$
(anti-rational in $\eta$).  The pole arguments show that neither is
identically zero; hence each has only finitely many zeros.
There are finitely many selected column sets.  Choose one $\eta$ outside the
finite set of parameter values at which any of these expressions, or any
coefficient of $p_\eta$, vanishes, and also outside the roots $\zeta^j$.
Then choose $\varepsilon$
sufficiently small.  We obtain independence for mixed $d$-tuples with
$s\in\{1,3,4,\ldots,d-3,d-1\}$.
This list is enough: from any $d+1$ vectors, delete one vector so that the
remaining $d$-tuple has one of these layer counts or is pure.  Thus every
$d+1$ vectors span, proving Part~(4) of \Cref{prop:odd-library}.

\section{Proof of Proposition~\ref{prop:odd-allocation}}\label{app:odd-allocation}

This appendix proves \Cref{prop:odd-allocation}.  The three cases below are
determined by the exceptional orthogonality graphs in
\Cref{prop:odd-library}.  In the \emph{single-non-qubit case}, there is no
ordinary non-qubit subsystem, so all remaining edges are assigned to
ordinary qubits.  Suppose instead that an ordinary non-qubit subsystem
exists.  If $d_*=5$, Part~(3) of \Cref{prop:odd-library} uses the same two
layer differences in both layers, and if $d_*\ge6$, Part~(4) uses the same
single layer difference in both layers.  In either situation, the same layer
differences remain available in both layers; this is the \emph{aligned
case}.  If $d_*\in\{3,4\}$, Parts~(1) and~(2) instead use $F_h^U$ in the
$U$ layer and $F_{h-1}^V$ in the $V$ layer.  The complement therefore
contains the mismatched pair $F_{h-1}^U$ and $F_h^V$.  In the
\emph{endpoint-absorption case}, we assign these two layer graphs to one
ordinary non-qubit subsystem, after which the remaining layer graphs are
aligned.
These cases are exhaustive because every non-qubit dimension is at least
three.

\subsection{A single non-qubit subsystem}

\begin{proof}[Single non-qubit case]
Let $d_*$ be the exceptional dimension.  All ordinary subsystems are
qubits, so it is enough to partition $\overline{G_*}$ into perfect matchings.

If $d_*=3$, use Part~(1) of \Cref{prop:odd-library}.  The exceptional
orthogonality graph
is $3$-regular, so its complement is $(2q-4)$-regular; the theorem of
Chetwynd and Hilton partitions it into perfect
matchings~\cite{ChetwyndHilton1985}.  If $d_*=4$, use Part~(2); the same
theorem partitions the $(2q-5)$-regular complement into perfect matchings.

If $d_*=5$, use Part~(3), and if $d_*\ge6$, use Part~(4).  In either case,
the complement is a Cayley graph of the dihedral group of order $2q$.
Recall that, for a group $\Gamma$ and a set $T\subseteq\Gamma$ closed under
inverses, the \emph{Cayley graph} has vertex set $\Gamma$ and joins $g$ to
$gt$ for every $g\in\Gamma$ and $t\in T$.  The \emph{dihedral group}
$D_{2q}$ consists of the $q$ rotations and $q$ reflections of a regular
$q$-gon; it is generated by a rotation $r$ and a reflection $s$, with
$r^q=s^2=1$ and $srs=r^{-1}$.  Here a set \emph{generates} a group if every
group element is a product of elements of the set and their inverses.  In our
graph, the unused cross-layer matchings include two consecutive shifts
$b,b+1$.
The corresponding generators satisfy
$r^{-1}=(r^bs)(r^{b+1}s)$ and hence generate $D_{2q}$.  A Cayley graph is
connected when its generators generate the underlying group, so the
complement is connected.  Stong's theorem therefore partitions its edges
into perfect matchings~\cite{Stong1985}.  In all cases there are exactly
$2q-d_*-1$ matchings, one for each ordinary qubit.
\end{proof}

\subsection{The aligned case}

For an integer $r\ge2$, define the set of allowed numbers of layer-difference
pairs by
\begin{equation}\label{eq:allowed-pairs}
 \mathcal A(r)=
 \begin{cases}
 \{0,1,\ldots,\lfloor r/2\rfloor\}\setminus\{r/4\},
 &4\mid r\text{ and }r\ge8,\\[1mm]
 \{0,1,\ldots,\lfloor r/2\rfloor\},&\text{otherwise}.
 \end{cases}
\end{equation}
The deleted value is exactly the solution of $2x=r-2x$, except that the
dimension-five case $r=4$ is retained because it has the separate biclique
bound in \Cref{lem:five-block}.

\begin{lemma}\label{lem:allowed-sums}
For integers $r,s\ge2$,
\[
 \mathcal A(r)+\mathcal A(s)
 =\{0,1,\ldots,\lfloor r/2\rfloor+\lfloor s/2\rfloor\}.
\]
Consequently, if $k\ge2$ and $r_1,\ldots,r_k\ge2$, then
\[
 \mathcal A(r_1)+\cdots+\mathcal A(r_k)
 =\{0,1,\ldots,\lfloor r_1/2\rfloor+\cdots+\lfloor r_k/2\rfloor\}.
\]
\end{lemma}

\begin{proof}
Write $u=\lfloor r/2\rfloor$ and $v=\lfloor s/2\rfloor$, and assume
$u\ge v$.  Each allowed set is either a full interval or a full interval
with its unique midpoint removed.  In particular, $\mathcal A(s)$ contains
$0$ and $1$.  Hence
$\mathcal A(r)+\{0,1\}$ contains every integer from $0$ to $u+1$: if the midpoint of $\mathcal A(r)$ is missing, it is obtained from the preceding integer plus one.  Both allowed sets are symmetric about their midpoints, so their sum is symmetric about $(u+v)/2$.  The reflection of $[0,u+1]$ is $[v-1,u+v]$, and the two intervals overlap.  This proves the first claim.  Adding further allowed sets preserves a full interval because every such set has gaps of size at most two.
\end{proof}

Assume that the remaining graph has the following form on two $q$-vertex
layers, where $q\ge7$ is odd:

\begin{itemize}
\item the same consecutive list of $L$ layer differences is available in
both layers, so the available layer graphs are aligned; if nonempty, this
list starts at $1$ or $2$ and lies in $\{1,\ldots,(q-1)/2\}$;
\item $C$ consecutive cross-layer matchings are available.
\end{itemize}

Let $e_1,\ldots,e_\ell$ be the dimensions of the ordinary non-qubit
subsystems, let $m$ be the number of ordinary qubits, and put
$S=\sum_{i=1}^{\ell}(e_i-1)$.
Assume $3\le e_i\le q-1$ for all $i$.  If the available list contains
$(q-1)/2$, assume also that every $e_i\ge6$.
Each aligned layer difference contributes degree two, while each cross-layer
matching contributes degree one.  Thus the required degree identity is
\begin{equation}\label{eq:aligned-total}
 2L+C=S+m.
\end{equation}
Let $\varepsilon$ be the number of even dimensions among
$e_1,\ldots,e_\ell$.
Define
\[
 C_0=
 \begin{cases}
 0,&m=0,\\
 \min\{m,C-\varepsilon\},&m>0,
 \end{cases}
\]
and
\begin{equation}\label{eq:aligned-X}
 X=\frac{S-C+C_0}{2}.
\end{equation}

\begin{lemma}\label{lem:aligned-allocation}
Suppose $C-\varepsilon\ge1$.
If $\ell=1$, suppose also that $X\in\mathcal A(e_1-1)$.  Then the remaining
graph can be partitioned into spanning graphs $G_i$ for the ordinary
non-qubit subsystems and $m$ perfect matchings for the ordinary qubits, such
that each $G_i$ is $(e_i-1)$-regular and $\bic(G_i)\le e_i$.
\end{lemma}

\begin{proof}
For each ordinary non-qubit dimension $e_i$, write $r_i=e_i-1$ and choose
$x_i\in\mathcal A(r_i)$; here $x_i$ is the number of layer-difference pairs
assigned to that subsystem.  By \eqref{eq:allowed-pairs}, a value is
deleted exactly when it would give the harmful equality
$2x_i=r_i-2x_i$ in a dimension at least nine; the dimension-five equality
is retained and handled by \Cref{lem:five-block}.

Because $S\equiv\varepsilon\pmod2$ and \eqref{eq:aligned-total} holds,
$C-\varepsilon\equiv m\pmod2$.  Thus $C_0\equiv m\pmod2$.  Moreover,
$C_0\ge C-S=m-2L$,
because both $m$ and $C-\varepsilon$ are at least $C-S$.  Thus $X$ is an
integer and $0\le X\le L$.
Since $2\sum_i\lfloor r_i/2\rfloor=S-\varepsilon$ and
$C_0\le C-\varepsilon$, we have
$X\le\sum_i\lfloor r_i/2\rfloor$.

If $\ell\ge2$, \Cref{lem:allowed-sums} allows us to choose
$x_i\in\mathcal A(r_i)$ with $\sum_i x_i=X$.
If $\ell=0$, then $X=0$ and there is nothing to choose.  If $\ell=1$,
take $x_1=X$, which belongs to $\mathcal A(r_1)$ by hypothesis.

Set $s_i=2x_i$ and $t_i=e_i-1-2x_i$.  Divide the layer differences into
consecutive cyclic annuli of $x_i$ pairs and the cross-layer matchings into
consecutive intervals of $t_i$ matchings.  A zero block is empty.  For a
nonempty layer block, \Cref{lem:odd-annulus} gives the biclique bound
$s_i+1$; for a cross-layer block, \Cref{lem:interval} gives $t_i+1$.  If
$s_i\ne t_i$, \Cref{lem:two-layer} yields
$\bic(G_i)\le s_i+t_i+1=e_i$.  The only retained equality is
$e_i=5$ and $s_i=t_i=2$.
The only $a\in\{1,\ldots,(q-1)/2\}$ satisfying
$2a\equiv\pm1\pmod q$ is $a=(q-1)/2$.  By hypothesis, this difference
is unavailable whenever a five-dimensional ordinary subsystem is present.
Thus \Cref{lem:five-block} gives $\bic(G_i)\le5$ in the retained equality
case as well.

Assign the annuli for the ordinary non-qubit subsystems from the high end of
the available consecutive list of layer differences.  The remaining layer
graphs therefore form an initial
interval.  Leave a consecutive interval of $C_0$ cross-layer matchings.  The
remaining graph $R$ has degree $2(L-X)+C_0=m$
by \eqref{eq:aligned-total} and \eqref{eq:aligned-X}.  It is a Cayley graph
of the dihedral group of order $2q$.  If $m=0$, it is empty, and if $m=1$,
it is already a perfect matching.  Assume $m\ge2$.  If $C_0\ge2$, two
consecutive reflection generators $r^bs,r^{b+1}s$ produce $r^{-1}$ and hence
generate $D_{2q}$, so $R$ is connected.  If $C_0=1$, layer edges remain,
and the remaining interval contains a difference $a\in\{1,2\}$.  Since $q$
is odd, $\gcd(a,q)=1$, so $r^a$ and the remaining reflection generate
$D_{2q}$.  Thus $R$ is again connected.

By Stong's theorem, the edges of $R$ can be partitioned into perfect
matchings~\cite{Stong1985}.  These perfect matchings are used for the
ordinary qubits.
\end{proof}

We now verify that the hypotheses of \Cref{lem:aligned-allocation} hold in
the aligned case.

\begin{proof}[Aligned case]
Assume that at least two subsystems are non-qubit and that their smallest
dimension $d_*$ is at least five.

Suppose first that $d_*=5$.  Part~(3) of \Cref{prop:odd-library} leaves
$L=(q-5)/2$ common layer-difference pairs and $C=q-1$ cross-layer
matchings.  The available layer differences are $1,\ldots,(q-5)/2$.
Every even ordinary non-qubit dimension is at least six, so
$5\varepsilon\le S\le2q-6$.  Hence $\varepsilon\le q-2=C-1$.

Now suppose $d_*\ge6$.  Part~(4) of \Cref{prop:odd-library} leaves
$L=(q-3)/2$ and $C=q-d_*+2$.  The available layer differences are
$2,\ldots,(q-1)/2$, and every ordinary non-qubit dimension is at least six.
Since $d_*$ is the smallest non-qubit
dimension, every even ordinary non-qubit dimension $e_i$ contributes at
least $d_*-1$ to $S$.  Thus
$\varepsilon(d_*-1)\le S\le2q-d_*-1$.
But
\[
 (C-1)(d_*-1)-(2q-d_*-1)
 =
 (q-d_*)(d_*-3)\ge0.
\]
Therefore $\varepsilon\le C-1$.

It remains to verify $X\in\mathcal A(D-1)$ when only one ordinary
non-qubit dimension $D$ remains.  The only possible failure occurs when
$X=L$ is the deleted midpoint, which would give $D-1=4L$.
For $d_*=5$, this gives $D=2q-9$.  The assumption $D\le q-1$ then gives
$q\le8$; since $q$ is odd and at least seven, $q=7$ and $D=5$.  This is not
a deleted value because $D-1=4$ is retained.  For $d_*\ge6$, the
equality gives $D=2q-5$, and $D\le q-1$ would force $q\le4$, impossible in
this branch.  Apply \Cref{lem:aligned-allocation}.
\end{proof}

\subsection{The endpoint-absorption case}

Assume that at least two subsystems are non-qubit and that their smallest
dimension $d_*$ belongs to $\{3,4\}$.  Let $q=2h+1\ge7$.  The exceptional
orthogonality graphs in Parts~(1) and~(2) of
\Cref{prop:odd-library} remove the layer graphs $F_h^U$ and $F_{h-1}^V$.
Their complement therefore contains the mismatched pair
\begin{equation}\label{eq:endpoint-E}
 E=F_{h-1}^U\mathbin{\dot\cup}F_h^V.
\end{equation}
Here $\dot\cup$ denotes an edge-disjoint union: the two graphs have the same
vertex set but disjoint edge sets.

\begin{lemma}\label{lem:endpoint-absorber}
Let $3\le d\le q-1$.  The graph in \eqref{eq:endpoint-E} can be extended to
a $(d-1)$-regular graph $A_d$ with $\bic(A_d)\le d$.
If $d\ne5$, take
\[
 A_d=
 E\mathbin{\dot\cup}
 \bigcup_{t\in I}M_t,
\]
where $I$ is a consecutive interval of $d-3$ cross-layer matchings.  If
$d=5$, take
\[
 A_5=
 E\mathbin{\dot\cup}F_{h-2}^U
 \mathbin{\dot\cup}F_{h-2}^V.
\]
\end{lemma}

\begin{proof}
Each of $F_{h-1}^U$ and $F_h^V$ is a union of odd cycles, so its biclique
number is at most three.  If $d\ne5$, the cross-layer interval has degree
$t=d-3$ and biclique bound $t+1$.  For $d=3$ there are no cross-layer edges.
For $d=4$ or $d\ge6$, the layer degree two differs from $t$, so
\Cref{lem:two-layer} gives $\bic(A_d)\le2+(d-3)+1=d$.

Let $d=5$.  In the $U$ layer the difference set
$\{\pm(h-2),\pm(h-1)\}$ is a four-element cyclic annulus, so
\Cref{lem:odd-annulus} gives biclique number at most five.  In the $V$ layer
the difference set is $S_0=\{\pm(h-2),\pm h\}$.
Suppose nonempty $X,Y\subseteq\Z_q$ satisfy $Y-X\subseteq S_0$.  Kneser's
theorem gives $|X|+|Y|\le|Y-X|+|H|$,
where $H$ is the stabilizer of $Y-X$.  If $H$ is trivial, the right side is
at most five.  If $H$ is nontrivial, then $|H|\ge3$.  Since $S_0$ has four
elements, the only possible case is $|H|=3$ and $Y-X$ is a three-element
coset.  After listing residues around the cycle $\Z_q$, their \emph{cyclic
gaps} are the forward distances between consecutive residues, including the
distance from the last residue back to the first.  For $S_0$, these gaps are
$2,1,2,q-5$.  No three of its residues have three equal cyclic gaps, so
$S_0$ contains no coset of a subgroup of order three.  Hence the stabilizer
is trivial.  There are no cross-layer edges in $A_5$, so
$\bic(A_5)\le5$.
\end{proof}

\begin{proof}[Endpoint-absorption case]
Choose an ordinary non-qubit subsystem of minimum dimension $d$ and assign it
the graph $A_d$ from \Cref{lem:endpoint-absorber}.  The remaining layer
graphs are now aligned.

Let $d_*\in\{3,4\}$ be the exceptional dimension.  If $d\ne5$, the
remaining parameters are $L=(q-5)/2$ and $C=q-d+2$ for $d_*=3$, or
$C=q-d+1$ for $d_*=4$.  If $d=5$, they are $L=(q-7)/2$ and $C=q-1$ for
$d_*=3$, or $C=q-2$ for $d_*=4$.
In both cases the available layer differences are $1,\ldots,L$ (empty
when $L=0$), so they exclude $(q-1)/2$ and satisfy the interval hypotheses
of \Cref{lem:aligned-allocation}.

Let $S$ be the sum of $e_i-1$ over the remaining ordinary non-qubit
subsystems, let $m$ be the number of ordinary qubits, and let $\varepsilon$
count the even dimensions among those ordinary non-qubit subsystems.  In
every case, $2L+C=S+m$.
We verify $C-\varepsilon\ge1$.

Assume first that $d\ne5$.  Every remaining even non-qubit dimension $e_i$
contributes at least $d-1$ to $S$.  When $d_*=3$,
\[
 (C-1)(d-1)-(S+m)
 =
 (d-3)(q-d)+2\ge0.
\]
When $d_*=4$,
\[
 (C-1)(d-1)-(S+m)
 =
 (d-3)(q-d)+4-d.
\]
This is nonnegative: the branch $d_*=4$ has $d\ge4$; the expression is $q-4$ for $d=4$, at least one for even $d\ge6$, and at least $d-2$ for odd $d\ge7$.

If $d=5$, every remaining even non-qubit dimension $e_i$ contributes at
least five to $S$.  Since $S+m$ is $2q-8$ or $2q-9$, we again obtain
$\varepsilon\le C-1$.

It remains to verify $X\in\mathcal A(D-1)$ when only one ordinary
non-qubit dimension $D$ remains.  If $d\ne5$, then $L=(q-5)/2$ and a
deleted midpoint would give $D=2q-9$.  The bound $D\le q-1$ leaves only
$q=7$ and $D=5$, for which $D-1=4$ is retained.  If $d=5$, then
$L=(q-7)/2$ and a deleted midpoint gives $D=2q-13$.  Together with
$5\le D\le q-1$, this leaves only $q=11$, $d=5$, and $D=9$.
In that case interchange the roles of the dimensions five and nine.  For the
remaining dimension five, $D-1=4$, and
$\mathcal A(4)=\{0,1,2\}$ has no deleted point.  Now apply
\Cref{lem:aligned-allocation} to complete the split.
\end{proof}

\section{A unified constructive proof for
\texorpdfstring{$(2,2,d)$}{(2,2,d)} with odd
\texorpdfstring{$d$}{d}}\label{app:theta}

The minimum $f_m(2,2,d)=d+3$ for odd $d\ge3$ is known:
Chen--Johnston~\cite[Theorem~3]{ChenJohnston2015} treated $d=4k+1$, and
Zhang et al.~\cite[Theorem~3.1]{ZhangEtAl2021} treated $d=4k-1$.
We retain the following proof because it is constructive and unified for
all odd $d\ge3$.  A single theta-graph construction gives $d+3$ product
states without separating the two congruence classes, and the
projected-incidence argument proves their unextendibility.  The general
projected-incidence construction below is also used for the small cases
in \Cref{sec:odd}.

\begin{theorem}\label{thm:22d}
For every odd integer $d\ge3$, the theta-graph construction below gives
a UPB of size $d+3$ in $\C^2\otimes\C^2\otimes\C^d$, and hence a
constructive proof of the known equality
\[
 f_m(2,2,d)=f_{LB}+1=d+3.
\]
\end{theorem}

\begin{proof}
The lower bound is \Cref{prop:lower}.  We first establish the projected-incidence construction used below.  Let $G$ be a simple graph on $N\ge3$ vertices and let $u$ be a distinguished vertex.  Orient all edges arbitrarily and write
\[
 g_e=\mathbf e_{e^-}-\mathbf e_{e^+}\in\R^{V(G)}
\]
for every edge $e$ with tail $e^-$ and head $e^+$.  Here $\mathbf e_v$ denotes the standard basis vector of $\R^{V(G)}$ corresponding to the vertex $v$: its $v$-coordinate is $1$ and all other coordinates are $0$.  The vectors $g_e\in\one^\perp$ are the incidence vectors of the oriented edges.  They are orthogonal exactly for disjoint edges.
Define
\[
 \rho=\mathbf e_u-\frac1N\one\in\one^\perp,
 \qquad
 W=\one^\perp\cap\rho^\perp.
\]
Let $P:\one^\perp\to W$ be the orthogonal projection onto $W$, that is,
\[
 P x=x-\frac{x\cdot\rho}{\rho\cdot\rho}\,\rho
 \qquad(x\in\one^\perp).
\]
Thus $P$ removes from $x$ its component parallel to $\rho$.  We make the following claim.
\begin{claim}\label{claim:1}
Every $Pg_e$ is nonzero and, for distinct edges,
\[
 Pg_e\perp Pg_f
 \quad\Longleftrightarrow\quad
 e\text{ and }f\text{ are disjoint in }G.
\]
Moreover, for $S\subseteq E(G)$,
\(
 \spann\{Pg_e:e\in S\}=W
\)
if and only if the spanning subgraph $(V(G),S)$ has at most two connected components.
\end{claim}

Indeed, since $g_e\perp\one$ and $\rho\cdot\rho=(N-1)/N$,
\begin{equation}\label{eq:projected-inner}
 (Pg_e)\cdot(Pg_f)
 =g_e\cdot g_f-\frac{N}{N-1}
  (g_e\cdot\mathbf e_u)(g_f\cdot\mathbf e_u).
\end{equation}
Disjoint edges give zero.  If $e$ and $f$ meet at $v\ne u$, then
$(Pg_e)\cdot(Pg_f)=g_e\cdot g_f=\pm1$.  If they meet at $u$, then we have
$(Pg_e)\cdot(Pg_f)=\pm (1-\frac{N}{N-1})$.
Moreover, $\|Pg_e\|^2$ is equal to $2$ if $u\notin e$ and to $(N-2)/(N-1)$ if $u\in e$, so every $Pg_e$ is nonzero.  Hence, any non-disjoint pair of $e,f$ gives $(Pg_e)\cdot(Pg_f)\ne0$, so the orthogonality condition is established.

If the spanning subgraph defined by \(S\) has \(h\) connected components,
then its incidence vectors \(g_e\) span the subspace consisting of all
vectors whose coordinates sum to zero on each component.
The incidence space, denoted by $U=\spann\{g_e:e\in S\}$, has dimension $N-h$.  For $h=1$, this is $\one^\perp$, whose projection is $W$. 
 For $h=2$, $U$ has dimension $N-2=\dim W$ and does not contain $\rho$: 
 the coordinate sum of $\rho$ on every nonempty proper vertex subset is nonzero.  
Since the kernel of \(P\) on \(\one^\perp\) is
\(\spann\{\rho\}\), the restriction of \(P\) to $U$ is injective, implying that $P(U)=W$.
  If $h\ge3$, $U$ has dimension at most $N-3$ and cannot project onto $W$.  This proves the claim.

Construct the graph \(\Theta\) from two distinct vertices \(u\) and \(v\) by
joining them with three internally vertex-disjoint paths of lengths
\(2,2,d-1\), where \(u\) is the distinguished vertex. 
Label the paths as in Figure~\ref{fig:theta-graph}.
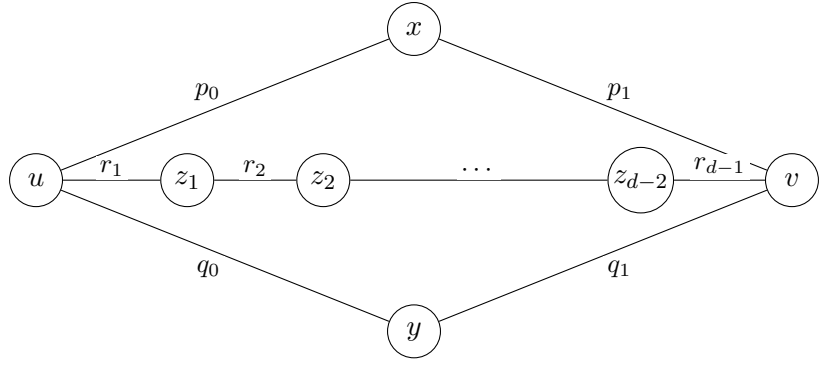
\begin{figure}[ht]
\centering
\begin{tikzpicture}[
  vertex/.style={circle,draw,fill=white,minimum size=7mm,inner sep=1pt},
  edge label/.style={fill=white,inner sep=1.5pt,font=\small}
]
  \node[vertex] (u) at (0,0) {$u$};
  \node[vertex] (v) at (10,0) {$v$};
  \node[vertex] (x) at (5,2) {$x$};
  \node[vertex] (y) at (5,-2) {$y$};
  \node[vertex] (z1) at (2,0) {$z_1$};
  \node[vertex] (z2) at (3.8,0) {$z_2$};
  \node[vertex] (zlast) at (8,0) {$z_{d-2}$};

  \draw (u) -- node[edge label,above left] {$p_0$} (x)
            -- node[edge label,above right] {$p_1$} (v);
  \draw (u) -- node[edge label,below left] {$q_0$} (y)
            -- node[edge label,below right] {$q_1$} (v);
  \draw (u) -- node[edge label,above] {$r_1$} (z1)
            -- node[edge label,above] {$r_2$} (z2);
  \draw (z2) -- node[edge label,above] {$\cdots$} (zlast)
              -- node[edge label,above] {$r_{d-1}$} (v);
\end{tikzpicture}
\caption{The graph $\Theta$: two $u$--$v$ paths of length $2$ and one of length $d-1$, with $u=z_0$ and $v=z_{d-1}$.  The diagram depicts $d\ge5$; for $d=3$, the long path is $u$--$z_1$--$v$.}
\label{fig:theta-graph}
\end{figure}
The graph has $d+2$ vertices and $f_{LB}+1=d+3$ edges.  
Since \(\dim W=d\), identify \(W\) isometrically with \(\R^d\).
Applying the projected-incidence construction to \(G=\Theta\) then
gives vectors \(Pg_e\in\R^d\) that are orthogonal exactly for disjoint edges.

Now we define $d+3$ product states in $\C^2\otimes\C^2\otimes\C^d$, one for each graph edge.  For $t\ge0$, put
\[
 \alpha_t=\begin{pmatrix}1\\t\end{pmatrix},
 \qquad
 \beta_t=\begin{pmatrix}t\\-1\end{pmatrix}.
\]
Then $\alpha_s\perp\beta_t$ exactly when $s=t$, and vectors of the same type are never orthogonal.  
Assign the first qubit by
\[
\begin{aligned}
 a_{p_0}=a_{q_1}&=\alpha_0,
 &a_{q_0}=a_{p_1}&=\beta_0,\\
 a_{r_{2j-1}}&=\alpha_j,
 &a_{r_{2j}}&=\beta_j
 &&\left(1\le j\le\frac{d-1}{2}\right),
\end{aligned}
\]
and the second qubit by
\[
\begin{aligned}
 b_{p_0}=b_{q_0}&=\alpha_0,
 &b_{r_1}&=\beta_0,\\
 b_{p_1}=b_{q_1}&=\alpha_{(d-1)/2},
 &b_{r_{d-1}}&=\beta_{(d-1)/2},\\
 b_{r_{2j}}&=\alpha_j,
 &b_{r_{2j+1}}&=\beta_j
 &&\left(1\le j\le\frac{d-3}{2}\right).
\end{aligned}
\]
For every graph edge $e$, define $\psi_e=a_e\otimes b_e\otimes(Pg_e)$.  
Disjoint graph edges are orthogonal in the third coordinate.  
At each graph vertex, the displayed qubit assignments cover every pair of incident edges.  
Hence the $d+3$ product states are pairwise orthogonal.

The proportionality classes in the first qubit are
\[
 \{p_0,q_1\},\quad\{q_0,p_1\},\quad
 \{r_1\},\ldots,\{r_{d-1}\},
\]
and those in the second qubit are
\[
 \{p_0,q_0\},\quad\{p_1,q_1\},\quad
 \{r_1\},\ldots,\{r_{d-1}\}.
\]
Deleting one class from each list leaves a spanning subgraph of $\Theta$ with at most two connected components.  Two deleted long edges isolate at most one middle segment; one short-path class and one long edge leave at most a $u$-side and a $v$-side; two short-path classes leave the long path intact and may isolate one short-path internal vertex.  By the spanning property proved above, the remaining $Pg_e$ span $\R^d$.

Suppose that a nonzero product vector \(a\otimes b\otimes c\) is
orthogonal to every \(\psi_e\). The edges satisfying
\(a\perp a_e\) are contained in one proportionality class of the first
list, and those satisfying \(b\perp b_e\) are contained in one class of
the second list. 
After deleting these two classes, the remaining vectors \(Pg_e\) span
\(\R^d\).  Orthogonality to all remaining states would therefore force
\(c=0\), a contradiction.
Thus the family of $\psi_e$ is a UPB of size $d+3$.
\end{proof}

%%%%%%%%%%%%%%%%%%%%%%%%%%%%%%%%%%%%%%%%%%%%%%%%%%%%%%%%%%

\end{document}